\documentclass[a4paper,UKenglish,cleveref, autoref, thm-restate]{lipics-v2021}

\pdfoutput=1

\title{A Temporal Logic for Strategic Hyperproperties} 

\author{Raven Beutner}{CISPA Helmholtz Center for Information Security, Germany} {}{}{}

\author{Bernd Finkbeiner}{CISPA Helmholtz Center for Information Security, Germany}{}{}{}

\authorrunning{R. Beutner and B. Finkbeiner} 

\Copyright{Raven Beutner and Bernd Finkbeiner} 

\ccsdesc[500]{Theory of computation~Modal and temporal logics}
\ccsdesc[500]{Theory of computation~Verification by model checking}

\keywords{hyperproperties, temporal logic, alternating-time temporal logic, model checking, multi-agent systems, information flow, asynchronous hyperproperties} 

\category{} 

\relatedversion{} 

\nolinenumbers 

\hideLIPIcs  

\usepackage{amsmath}
\usepackage{amsfonts}
\usepackage{amssymb}
\usepackage{ wasysym }
\usepackage{amsthm}
\usepackage{ stmaryrd }
\usepackage{MnSymbol}
\usepackage{listings}
\usepackage{tikz}
\usepackage{stackengine}
\usepackage{tabularx}
\usepackage{bussproofs}

\usepackage{wrapfig}
\usepackage{caption}
\usepackage{subcaption}
\usepackage{makecell}

\usepackage{pifont}
\newcommand{\cmark}{\ding{51}}%
\newcommand{\xmark}{\ding{55}}%

\newcommand{\EXPTIME}{\texttt{EXPTIME}}
\newcommand{\EXPSPACE}{\texttt{EXPSPACE}}

\newcommand{\PSPACE}{\texttt{PSPACE}}

\newcommand{\NLOGSPACE}{\texttt{NLOGSPACE}}

\newcommand{\ltlnext}{\tikz[baseline]{
		\draw[line width=.12ex]
		(0,.6ex) circle (.8ex);
}}
\newcommand{\ltlg}{\tikz[baseline]{
		\draw[line width=.12ex,line join=round]
		(0ex,-.2ex) -- (0ex,1.3ex) -- (1.5ex,1.3ex) -- (1.5ex,.-.2ex) -- cycle;
}}

\newcommand{\ltlf}{\tikz[baseline]{
		\draw[line width=.12ex,line join=round]
		(0ex,.6ex) -- (.95ex,1.55ex) -- (1.9ex,.6ex) -- (.95ex,-.35ex) -- cycle;
}}

\newcommand{\HyperATLS}{\texttt{HyperATL}$^*$}

\newcommand{\CTL}{\texttt{CTL}}
\newcommand{\CTLS}{\texttt{CTL}$^*$}

\newcommand{\LTL}{\texttt{LTL}}
\newcommand{\HyperLTL}{\texttt{HyperLTL}}

\newcommand{\ATL}{\texttt{ATL}}
\newcommand{\ATLS}{\texttt{ATL}$^*$}

\newcommand{\AHLTL}{\texttt{AHLTL}}

\newcommand{\HyperCTLS}{\texttt{HyperCTL}$^*$}

\newcommand{\atomic}{\mathbf{AP}}

\newcommand{\agents}{\Xi}
\newcommand{\agent}{\xi}

\newcommand{\refLemma}[1]{Lemma~\ref{lem:#1}}
\newcommand{\refTheo}[1]{Theorem~\ref{theo:#1}}
\newcommand{\refDef}[1]{Definition~\ref{def:#1}}
\newcommand{\refFig}[1]{Fig.~\ref{fig:#1}}
\newcommand{\refSection}[1]{Sec.~\ref{sec:#1}}

\newcommand{\refTable}[1]{Table~\ref{tab:#1}}
\newcommand{\refProp}[1]{Proposition~\ref{prop:#1}}

\definecolor{controlColor}{RGB}{0,0, 205}
\definecolor{varColor}{RGB}{128,0,128}
\definecolor{constColor}{RGB}{0,139,0}

\newcommand{\myvar}[1]{{\ttfamily \color{varColor} #1}}
\newcommand{\mycontrol}[1]{{\ttfamily \color{controlColor} #1}}
\newcommand{\myconst}[1]{{\ttfamily \color{constColor} #1}}

\begin{document}

\maketitle

\begin{abstract}
Hyperproperties are commonly used in computer security to define information-flow policies and other requirements that reason about the relationship between multiple computations. In this paper, we study a novel class of hyperproperties where the individual computation paths are chosen by the strategic choices of a coalition of agents in a multi-agent system. We introduce \HyperATLS, an extension of computation tree logic with path variables and strategy quantifiers. \HyperATLS~can express strategic hyperproperties, such as that the scheduler in a concurrent system has a \emph{strategy} to avoid information leakage. \HyperATLS~is particularly useful to specify \emph{asynchronous} hyperproperties, i.e., hyperproperties where the speed of the execution on the different computation paths depends on the choices of the scheduler.  Unlike other recent logics for the specification of asynchronous hyperproperties, our logic is the first to admit decidable model checking for the full logic. We present a model checking algorithm for \HyperATLS~based on alternating word automata and show that our algorithm is asymptotically optimal by providing a matching lower bound. 
We have implemented a prototype model checker for a fragment of \HyperATLS, able to check various security properties on small programs.
\end{abstract}

\section{Introduction}

Hyperproperties \cite{DBLP:journals/jcs/ClarksonS10} are system properties that specify a relation between the traces of the system.
Such properties are of increasing importance as they can, for example, characterize the information-flow in a system \cite{DBLP:phd/dnb/Rabe16}.
Consequently, several logics for the specification of hyperproperties have been developed, including hyper variants of \CTLS (and \LTL)~\cite{DBLP:conf/post/ClarksonFKMRS14,DBLP:phd/dnb/Rabe16}, \texttt{PDL}-$\Delta$~\cite{DBLP:conf/concur/GutsfeldMO20} and \texttt{QPTL}~\cite{DBLP:conf/cav/FinkbeinerHHT20}.  
A prominent example is the temporal hyperlogic \HyperLTL~\cite{DBLP:conf/post/ClarksonFKMRS14}, which extends linear-time temporal logic (\LTL)~\cite{DBLP:conf/focs/Pnueli77} with explicit trace quantification. 
In \HyperLTL~we can, for instance, express non-interference (NI), i.e., the requirement that the observable output of system does not depend on high-security inputs~\cite{DBLP:conf/sp/GoguenM82a}.
A prominent formulation of NI for non-deterministic systems is \emph{generalized non-interference} (GNI)~\cite{DBLP:conf/sp/McCullough88,DBLP:conf/cav/CoenenFST19}, which can be expressed as the \HyperLTL~formula
$$\forall \pi_1.~\forall \pi_2.~\exists \pi_3.~\ltlg(\bigwedge_{a \in H} a_{\pi_1} \leftrightarrow a_{\pi_3}) \land \ltlg(\bigwedge_{a \in O} a_{\pi_2} \leftrightarrow a_{\pi_3}),$$
where $H$ and $O$ are two sets of propositions, with $H$ representing the high-security input and $O$ the output.
The formula states that for any pair of traces $\pi_1, \pi_2$ there exists a third trace that agrees on the high-security inputs with $\pi_1$ and on the outputs with $\pi_2$ (for simplicity we assume that no low-inputs are present).
The existence of such a trace guarantees that any observation made on the outputs is compatible with every possible sequence of high-security inputs.
The non-determinism is thus the sole explanation for the system output.

In this paper, we introduce a novel class of hyperproperties that reason about \emph{strategic behavior} in a multi-agent system. 
As a motivation for why strategic hyperproperties are desirable, consider $\mathit{GNI}$ from above.
As \HyperLTL~only quantifies existentially or universally over the paths in the system, the entire system is treated  either as fully controllable or fully adversarial.
Moreover, the witness trace $\pi_3$ can be constructed with full knowledge of both $\pi_1$ and $\pi_2$; this means that the entire output and input history can be used to resolve the non-determinism of the system appropriately. 
Now consider a system where the non-determinism arises from a scheduling decision between two possible subprograms $P_1, P_2$.
Each subprogram reads the next input \myvar{h} of the system. Suppose that $P_1$ assumes that \myvar{h} is even and otherwise leaks information, while $P_2$ assumes that \myvar{h} is odd and otherwise leaks information.
In the trace-based view of $\mathit{GNI}$, the witness trace $\pi_3$ is fixed knowing the entire \emph{future} input sequence, allowing the construction of a leakage-avoiding path $\pi_3$; The system satisfies $\mathit{GNI}$.
An \emph{actual} scheduler, who chooses which of $P_1, P_2$ handles the next input, can only avoid a leakage if it knows what the \emph{next} input will be, which is impossible in a real-world system.
The \HyperLTL~formulation of $\mathit{GNI}$ is, in this case, unable to express the desired property.  
In our scenario, we need to reason about the \emph{strategic} behaviour of the system, i.e., we want to check if there exist a strategy for the scheduler that avoids leakage.

\subparagraph{Strategic Hyperproperties}
Reasoning about strategic behavior in multi-agent systems has been studied before. The seminal work on alternating-time temporal logic~\cite{DBLP:journals/jacm/AlurHK02} introduced an extension of \CTL~(and \CTLS \cite{DBLP:journals/jacm/EmersonH86}) that is centred around the idea of viewing paths as the outcome of a game, where some agents are controlled via a strategy. 
The \ATLS~quantifier $\llangle A \rrangle \varphi$ requires the agents in $A$ to have a strategy that enforces the path formula $\varphi$ to become true. 
This makes \ATLS~an ideal logic for reasoning about \emph{open} systems, where one is less interested in the pure existence of a path, but rather in the actually realizability of an outcome in a multi-agent system.
\ATL~has numerous variations and extensions, which, for example, introduce knowledge modalities \cite{DBLP:journals/sLogica/HoekW03a} or imperfect observation \cite{DBLP:conf/atal/BerthonMM17}. 
While strategy quantifiers in \ATLS~can be nested (like in \CTLS), the logic is still unable to express hyperproperties, as the scope of each quantifier ends with the beginning of the next (see \cite{DBLP:journals/eatcs/Finkbeiner17}).

It is very useful to reason about the strategic behaviour of the agents in a multi-agent system with respect to a hyperproperty. 
In the example above, one would like to ask if the scheduler has a \emph{strategy} (based on the finite history of inputs only) such that unintended information-flow (which is a hyperproperty) is prevented (in the above example such an answer should be negative). 
There exist multiple angles to approach this:
One could, for instance, interpret strategic hyperproperties such that a coalition of agents tries to achieve a set of outcomes satisfying some hyperproperty (expressed, for example,  in \HyperLTL). 
Model checking the resulting logic would then subsume realizability of \HyperLTL, which is \emph{undecidable} even for simple alternation-free formulas \cite{DBLP:conf/cav/FinkbeinerHLST18}.

In this paper, we introduce a new temporal logic, called \HyperATLS, that combines the strategic behaviour in multi-agent systems with the ability to express hyperproperties. 
Crucially, we focus on the strategic behaviour of a coalition of agents along a \emph{single path}, i.e., we view path quantification as the outcome of a game. 
Syntactically, we follow a similar approach as alternating-time temporal logic~\cite{DBLP:journals/jacm/AlurHK02}. 
We use the strategy quantifier $\llangle A \rrangle \pi. \varphi$ to specify that the agents in $A$ have a strategy such that each possible outcome, when bound to the path variable $\pi$, satisfies $\varphi$.
A formula of the from $\llangle A_1 \rrangle \pi_1.  \llangle A_2 \rrangle \pi_2. \varphi$ now requires the existence of strategy for the agents in $A_1$ such that for all possible outcomes of the game $\pi_1$, the agents in $A_2$ have a strategy such that for all possible outcomes $\pi_2$, the combination of $\pi_1, \pi_2$ satisfies $\varphi$ (which is a formula that can refer to propositions on paths $\pi_1, \pi_2$).
The strategic behaviour chosen by each quantifier is thus limited to the current path and can be based on the already \emph{fixed} outcomes of outer quantifiers (i.e., the entire strategy for the agents in $A_2$ can depend on the full outcome of $\pi_1$).
Sometimes, however, it is useful not to reason incrementally about the strategy for a single path at a time, but rather to reason about a \emph{joint} strategy for multiple paths.
To express this, we endow our logic with an explicit construct to resolve the games in parallel (syntactically we surround quantifiers by $[\cdot]$ brackets).
The formula $[\llangle A_1 \rrangle \pi_1.  \llangle A_2 \rrangle \pi_2.]~\varphi$ requires winning strategies for the agents in $A_1$ (for the first copy) and for $A_2$ (for the second copy) where the strategies can observe the current state of \emph{both} copies.
This enables collaboration between the agents in $A_1$ and $A_2$.

Similar to \ATLS, the empty (resp.~full) agent set corresponds to universal (resp.~existential) quantification. 
\HyperATLS~therefore subsumes \HyperCTLS~(and thus \HyperLTL) as well as \ATLS.
The logic is thus a natural extension of both the temporal logics for hyperproperties and the alternating-time logics from the non-hyper realm (see \refFig{expr}).

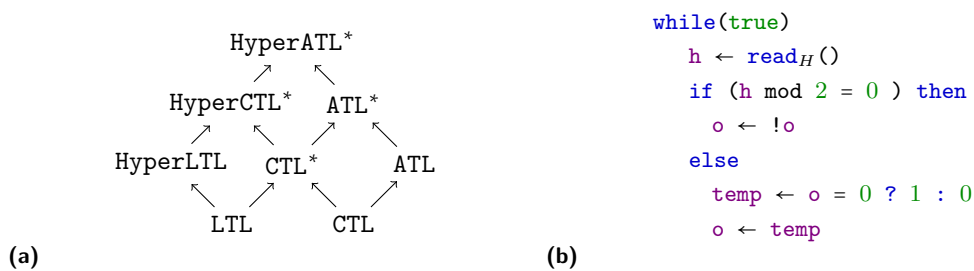
\begin{figure}
	\begin{subfigure}[b]{0.5\textwidth}
		\begin{center}
			\begin{tikzpicture}[scale=0.8]
				
				\node[] at (1,3) (hatl) {\HyperATLS};

				\node[] at (0,2) (hctl) {\HyperCTLS};
				
				\node[] at (2,2) (atls) {\ATLS};
				
				\node[] at (-1,1) (hltl) {\texttt{HyperLTL}};
				
				\node[] at (1,1) (ctls) {\CTLS};
				
				\node[] at (3,1) (atl) {\texttt{ATL}};

				\node[] at (0,0) (ltl) {\texttt{LTL}};
				
				\node[] at (2,0) (ctl) {\texttt{CTL}};

				\draw[->] (ltl) -- (ctls);
				\draw[->] (ltl) -- (hltl);
				
				\draw[->] (ctl) -- (ctls);
				\draw[->] (ctl) -- (atl);
				
				\draw[->] (hltl) -- (hctl);
				
				\draw[->] (ctls) -- (hctl);
				\draw[->] (ctls) -- (atls);
				
				\draw[->] (atl) -- (atls);
				
				\draw[->] (hctl) -- (hatl);
				\draw[->] (atls) -- (hatl);
				
			\end{tikzpicture}
		\end{center}
		\vspace{-0.5cm}
		\subcaption{ }\label{fig:expr}
	\end{subfigure}
	\begin{subfigure}[b]{0.5\textwidth}
		\begin{center}
			\begin{minipage}{0.6\textwidth}
				{\small\ttfamily
					\textcolor{controlColor}{while}(\textcolor{constColor}{true})\\
					\makebox[0.3cm]{} \myvar{h} $\leftarrow$ \textcolor{controlColor}{read}$_H$()\\
					\makebox[0.3cm]{} \textcolor{controlColor}{if} (\myvar{h} mod \textcolor{constColor}{$2$}~$=$~\textcolor{constColor}{$0$} ) \textcolor{controlColor}{then}\\
					\makebox[0.6cm]{} \myvar{o} $\leftarrow$ !\myvar{o}\\
					\makebox[0.3cm]{} \textcolor{controlColor}{else}\\
					\makebox[0.6cm]{} \textcolor{varColor}{temp} $\leftarrow$ \myvar{o}~$ = $~\myconst{$0$} \mycontrol{?}~\textcolor{constColor}{$1$} \mycontrol{:}~\textcolor{constColor}{$0$} \\
					\makebox[0.6cm]{} \myvar{o} $\leftarrow$ \textcolor{varColor}{temp}	
				}
			\end{minipage}
		\end{center}
		\vspace{-0.5cm}
		\subcaption{}\label{fig:exProg}
	\end{subfigure}
	\caption{\textbf{(a)}: Expressiveness of temporal logics. An arrow $A \to B$ indicates that $A$ is a syntactic fragment of $B$. \textbf{(b)}: Example program that violates (synchronous) observational-determinism. }
\end{figure}

\subparagraph{Strategic Non-Interference}
Consider again the example of $\mathit{GNI}$ expressed in \HyperLTL.
In \HyperATLS, we can express a more refined, strategic notion of non-interference, that postulates the existence of a \emph{strategy} for the non-determinism. 
As a first step, we consider a program no longer as a Kripke structure (a standard model for temporal hyperlogics), but as a game structure played between two players. Player $\agent_N$ is responsible for resolving the non-determinism of the system, and player $\agent_H$ is responsible for choosing the high-security inputs to the system.
We can now express that $\agent_N$ has a strategy to produce matching outputs (without knowing the future inputs by $\agent_H$). 
Consider the following formula $\mathit{stratNI}$:
$$\forall \pi_1.~\llangle \{\agent_N\} \rrangle \pi_2.~\ltlg (\bigwedge_{a \in O} a_{\pi_1} \leftrightarrow a_{\pi_2})$$
This formula requires that for every possible reference path $\pi_1$, the non-determinism always has a \emph{strategy} to produce identical outputs.
One can show that $\mathit{stratNI}$ implies $\mathit{GNI}$: The existence of a leakage ``disproving'' strategy implies the existence of a leakage ``disproving'' trace. 
A particular strength of this formulation is that we can encode additional requirements on the strategy. 
For example: if the internal non-determinism arises from the scheduling decisions between multiple components, we can require fairness of the scheduling strategy. 

\subparagraph{Asynchronous Hyperproperties}
Strategic hyperproperties are also very natural to express \emph{asynchronous hyperproperties}.
While existing hyperlogics traverse the traces of a system synchronously, one often requires an asynchronous traversal to account, for example, for the unknown speed of execution of software that runs on some unkown platform.  
In a multi-agent system, the scheduling decision (i.e., whether a system progresses or remains in its current state) can then be seen as the decision made by scheduling agent (called $\mathit{sched}$ in the following).
If not already present, we can artificially add such a scheduling agent via a system transformation. 
By either including or excluding this agent in a strategy quantifier, we can then naturally reason about asynchronous executions of programs.
Instead of reasoning about the asynchronous scheduling of a system directly, we thus reason about the existence of a strategy for the scheduling agent. 

As an example consider the program in \refFig{exProg}, which continuously reads an input and flips the output \myvar{o} either directly, or via a temporary variable.
Based on the input, the exact timepoint of the change in \myvar{o} differs.
A synchronous formulation of \emph{observational-determinism} (OD) \cite{DBLP:conf/csfw/HuismanWS06}, which requires the output to be identical on all traces, does not hold.
In \HyperATLS, we can naturally express a variant of OD where we search for a strategy for the scheduling agent $\mathit{sched}$, who aligns the outputs on both traces by stuttering them appropriately:
$$[\llangle \{\mathit{sched}\} \rrangle  \pi_1.~\llangle \{\mathit{sched}\} \rrangle \pi_2.]~\ltlg (\bigwedge_{a \in O} a_{\pi_1} \leftrightarrow a_{\pi_2}) $$
The program in \refFig{exProg} (with an added asynchronous scheduler) satisfies this variant, because $\mathit{sched}$ can stutter the change in \myvar{o} in order to align with the second trace.

To demonstrate the expressiveness of this strategic view on asynchronous hyperproperties, we compare our approach to \AHLTL, a recent temporal logic for asynchronous hyperproperties~\cite{hyperaltl}.
While \AHLTL~model checking is undecidable in general, recent work~\cite{hyperaltl} has identified a large fragment for which model checking is possible. 
We show that this fragment can be encoded within \HyperATLS.
Every property in this (largest known) decidable fragment can thus be expressed in \HyperATLS, for which model checking is decidable for the \emph{full} logic.

\subparagraph{Model Checking}
We show that model checking of \HyperATLS~on concurrent game structures is decidable and present an automata-theoretic algorithm. 
Our algorithm incrementally reduces model checking to the emptiness of an automaton.
We show that alternating automata are well suited to keep track of all possible path assignments satisfying a formula by encoding the game structure in the transition function of the automaton.
We characterize the model checking complexity in terms of the number of complex quantifiers (where the agent team is non-trivial) and simple quantifiers (i.e., $\exists$ or $\forall$).
We provide a lower bound, based on a novel construction that encodes a doubly exponential counter within a single strategy quantifier, that (in almost all cases) matches the upper bound from our algorithm.

\subparagraph{Prototype Model Checker}
On the practical side, we present a prototype model checker for an efficient fragment of \HyperATLS~by reducing the model checking to solving of a parity game. 
The fragment supported by our tool does, in particular, include all alternation free \HyperLTL~formulas \cite{DBLP:conf/cav/FinkbeinerRS15}, the $\forall^*\exists^*$-model checking approach from \cite{DBLP:conf/cav/CoenenFST19} as well as all formulas in the decidable fragment of \AHLTL~\cite{hyperaltl}.

\subparagraph{Contributions}
In summary, our contributions include the following:
\begin{itemize}
	\item We introduce a novel logic to express strategic hyperproperties and demonstrate that it is well suited to express, e.g.,~information-flow control and, in particular, asynchronous hyperproperties. 
	
	\item We give an automata-based model checking algorithm for our logic and provide a lower bound on the model checking problem.	
	
	\item We show that our logic can express all formulas in the largest known decidable fragment of the existing hyperlogic \AHLTL~\cite{hyperaltl}.

	\item We provide a prototype-model checker for an efficiently checkable fragment of \HyperATLS~and use it to verify information-flow polices and asynchronous hyperproperties.
	
\end{itemize}

\section{Preliminaries}

In this section we introduce some basic preliminaries needed in the following.

\subparagraph{Concurrent Game Structure}
As our model of multi-agent systems, we consider concurrent game structures (CGS) \cite{DBLP:journals/jacm/AlurHK02}.
The transition relation in a CGS is based on the decision by individual agents (or players).
Formally, a CGS is a tuple $\mathcal{G} = (S, s_0, \agents, \mathscr{M}, \delta, \atomic, L)$ where $S$ is a finite set of states, $s_0 \in S$ the initial state, $\agents$ a finite set of agents and $\mathscr{M}$ a finite set of moves.
We call a function $\sigma: \agents \to \mathscr{M}$ a global move vector and for a set of agent $A \subseteq \agents$ a function $\sigma: A \to \mathscr{M}$ a partial move vector. 
$\delta: S \times (\agents \to \mathscr{M}) \to S$ is a transition function that maps states and move vectors to successor states. 
Finally, $\atomic$ is a finite set of propositions and $L : S \to 2^\atomic$ a labelling function. 
Note that every Kripke structure (a standard model for temporal logics \cite{DBLP:books/daglib/0020348}) can be seen as a $1$-player CGS.
For disjoint sets of agents $A_1, A_2$ and partial move vectors $\sigma_i : A_i \to \mathscr{M}$ for $i \in \{1, 2\}$ we define $\sigma_1 + \sigma_2$ as the move vector obtained as the combination of the individual choices. 
For $\sigma : A \to \mathscr{M}$ and $A' \subseteq A$, we define $\sigma_{\mid A'}$ as the move vector obtained by restring the domain of $\sigma$ to $A'$.

In a concurrent game structure (as the name suggests) all agents choose their next move concurrently, i.e., without knowing what moves the other player have chosen. 
We introduce the concept of multi-stage CGS (MSCGS), in which the move selection proceeds in stages and agents can base their decision on the already selected moves of (some of the) other agents. 
This is particularly useful when we, e.g.,~want to base a scheduling decision on the moves selected by the other agents.
Formally, a MSCGS is a CGS equipped with a function $d : \agents \to \mathbb{N}$, that orders the agents according to informedness. 
Whenever $d(\agent_1) < d(\agent_2)$, $\agent_2$ can base its next move on the move selected by $\agent_1$.
A CGS thus naturally corresponds to a MSCGS with $d = \mathbf{0}$, where $\mathbf{0}$ is the constant $0$ function.

\subparagraph{Alternating Automata}
An alternating parity (word) automaton (APA) is a tuple $\mathcal{A} = (Q, q_0, \Sigma, \rho, c)$ where $Q$ is a finite set of states, $q_0$ an initial state, $\Sigma$ a finite alphabet, $\rho : Q \times \Sigma \to \mathbb{B}^+(Q)$ a function mapping states to positive boolean combinations of states and $c : Q \to \mathbb{N}$ a colouring of nodes with natural numbers \cite{DBLP:journals/jacm/ChandraKS81}. 
For $\varphi \in \mathbb{B}^+(Q)$, $B \subseteq Q$ we write $B \models \varphi$ if the assignment obtained from $B$ satisfies $\varphi$.
A tree is a set $T \subseteq \mathbb{N}^*$ that is prefixed closed, i.e., $\tau \cdot n \in T$ implies $\tau \in T$.
We refer to elements in $\tau \in T$ as nodes and denote with $|\tau|$ the length of $\tau$ (or equivalently the depth of the node). 
For a node $\tau \in T$ we denote with $\mathit{children}(\tau)$ the immediate children of $\tau$, i.e., $\mathit{children}(\tau) = \{\tau \cdot n \in T \mid n \in \mathbb{N}\}$.
An $X$-labelled tree is a pair $(T, r)$ where $T$ is a tree and $r : T \to X$ a labelling with $X$. 
A run of an APA $\mathcal{A} = (Q, q_0, \Sigma, \rho, c)$ on a word $u \in \Sigma^\omega$ is a $Q$-labelled tree $(T, r)$ that satisfies the following: \textbf{(1)} $r(\epsilon) = q_0$, \textbf{(2)} For all $\tau \in T$, $\{r(\tau') \mid \tau' \in \mathit{children}(\tau)\} \models \rho(r(\tau), u(|\tau|))$.
A run $(T, r)$ is accepting if for every infinite path $\pi$ in $T$ the minimal colour (given by $c$) that occurs infinitely many times is even. 
We denote with $\mathcal{L}(\mathcal{A})$ the set of words for which $\mathcal{A}$ has an accepting run.
We call an alternating automaton $\mathcal{A}$ non-deterministic (resp.~universal) if the transition function $\delta$ is a disjunction (resp.~conjunction) of states.
If $\delta$ is just a single state, we call $\mathcal{A}$ deterministic.
Crucially alternating, non-deterministic, universal and deterministic parity automaton are all equivalent in the sense that they accept the same class of languages (namely $\omega$-regular ones) although they can be (double) exponentially more succinct:

\begin{theorem}[\cite{DBLP:journals/tcs/MiyanoH84,DBLP:journals/jacm/DrusinskyH94}]\label{theo:alt1}
	For every alternating parity automaton $\mathcal{A}$ with $n$ states, there exists a non-deterministic parity automaton $\mathcal{A}'$ with $2^{\mathcal{O}(n)}$-states that accepts the same language.
	For every non-deterministic or universal parity automaton $\mathcal{A}$ with $n$ states, there exists a deterministic parity automaton $\mathcal{A}'$ with $2^{\mathcal{O}(n)}$-states that accepts the same language.
\end{theorem}

\begin{theorem}[\cite{DBLP:journals/tocl/KupfermanV01}]\label{theo:altneg}
	For every alternating parity automaton $\mathcal{A}$ with $n$ states, there exists a alternating parity automaton $\overline{\mathcal{A}}$ with $\mathcal{O}(n^2)$-states that accepts the complemented language.
\end{theorem}

\section{HyperATL*}\label{sec:hypteratl}

In this section we introduce \HyperATLS. 
Our logic extends the standard temporal logic \CTLS~\cite{DBLP:journals/jacm/EmersonH86} by introducing path variables and strategic quantification \cite{DBLP:journals/jacm/AlurHK02}. 
Assume a countably infinite set of path variables $\mathit{Var}$, a set of agents $\agents$ and a set of atomic propositions $\atomic$.
\HyperATLS~formulas are generated by the following grammar
$$ \varphi := \llangle A \rrangle \pi. \varphi \mid a_\pi \mid \varphi \land \varphi \mid \neg \varphi \mid \varphi~\mathcal{U}~\varphi \mid \ltlnext \varphi$$
where $\pi \in \mathit{Var}$ is a path variable, $a \in \atomic$ an atomic proposition and $A \subseteq \agents$ a set of agents.
As in \HyperCTLS, $a_\pi$ means that proposition $a$ holds in the current step on path $\pi$. Via $\llangle A \rrangle \pi. \varphi$ we can quantify over paths in a system (which we consider as the outcome of a game). 
$\llangle A \rrangle \pi. \varphi$ requires the agents in $A$ to have a \emph{strategy} (defined below) such that each outcome under that strategy, when bound to trace variable $\pi$, satisfies $\varphi$.
We abbreviate as usual $\varphi_1 \lor \varphi_2 := \neg (\neg \varphi_1 \land \neg \varphi_2)$, and the temporal operators globally ($\ltlg$), eventually ($\ltlf$) and release ($\mathcal{R}$).
Trivial agent sets, i.e., $A= \emptyset$ or $A = \agents$ correspond to classical existential or universal quantification.
We therefore write $\forall \pi$ instead of $\llangle \emptyset \rrangle \pi$ and $\exists \pi$ instead of $\llangle \agents \rrangle \pi$.  
We call a quantifier \emph{simple} if the agent-set is trivial and otherwise \emph{complex}.
We call a formula \emph{linear} if it consists of an initial quantifier prefix followed by a quantifier-free (\LTL) formula.

\subparagraph{Semantics}
Let us fix a MSCGS $\mathcal{G} = (S, s_0, \agents, \mathscr{M}, \delta, d, \atomic, L)$ as a model.
We first need to formalize the notion of a strategy in the game structure. A strategy for any agent is a function that maps finite history of plays in the game to a move in $\mathscr{M}$. 
As the plays in an MSCGS progress in stages the decision can be based not only on the past sequence of states, but also on the fixed moves of all agents in previous stages.
Formally, a strategy for an agent $\agent$ is a function $f_\agent : S^+ \times (\{\agent' \mid d(\agent') < d(\agent)\} \to \mathscr{M}) \to \mathscr{M}$.
Given a set of agents $A$, a set of strategies $F_A = \{f_\agent \mid \agent \in A\}$ and a state $s \in S$, we define $\mathit{out}(\mathcal{G}, s, F_A)$ as the set of all runs $u \in S^\omega$ such that \textbf{(1)} $u(0) = s$ and \textbf{(2)} for every $i \in \mathbb{N}$ there exists a global move vector $\sigma$ with $\delta(u(i), \sigma) = u(i+1)$ and for all $\agent \in A$ we have $\sigma(\agent) = f_\agent(u[0, i], \sigma_{\mid  \{\agent' \mid d(\agent') < d(\agent)\}})$.
The agents in $A$ choose their move in each step based on the finite history of the play and the decision of all other agents in an earlier stage.
Note that in case where $d = \mathbf{0}$, a strategy is just a function $S^+ \to \mathscr{M}$, ignoring the moves selected by other agents. 

The semantics of a formula is now defined in terms of a path assignment $\Pi : \mathit{Var} \to S^\omega$, mapping path variables to infinite sequences of states in $\mathcal{G}$.
For a path $t \in S^\omega$ we write $t[i, \infty]$ to refer to the infinite suffix of $t$ starting at position $i$. 
We write $\Pi[i, \infty]$ to denote the path assignment defined by $\Pi[i, \infty](\pi) = \Pi(\pi)[i, \infty]$. 
We can then inductively define the satisfaction relation for \HyperATLS:
\begin{align*}
	\Pi &\models_\mathcal{G} a_\pi &&\text{iff  } a \in L(\Pi(\pi)(0))\\
	\Pi &\models_\mathcal{G} \neg \varphi &&\text{iff  }\Pi \not\models_\mathcal{G} \varphi\\
	\Pi &\models_\mathcal{G} \varphi_1 \land \varphi_2 &&\text{iff  } \Pi \models_S \varphi_1 \text{ and } \Pi \models_\mathcal{G} \varphi_2\\
	\Pi &\models_\mathcal{G} \varphi_1~\mathcal{U}~\varphi_2 &&\text{iff  } \exists i \geq 0.  \Pi[i, \infty] \models_\mathcal{G} \varphi_2  \text{ and } \forall 0 \leq j < i. \Pi[j, \infty] \models_\mathcal{G} \varphi_1\\
	\Pi &\models_\mathcal{G} \ltlnext \varphi &&\text{iff  }\Pi[1, \infty] \models_\mathcal{G} \varphi\\
	\Pi &\models_\mathcal{G} \llangle A \rrangle \pi.~\varphi &&\text{iff  } \exists F_A: \forall t \in \mathit{out}(\mathcal{G}, \Pi(\epsilon)(0), F_A): \Pi[\pi \mapsto t] \models_\mathcal{G} \varphi
\end{align*}
Here $\Pi(\epsilon)$ refers to the path that was last added to the assignment (similar to the  \HyperLTL-semantics \cite{DBLP:conf/post/ClarksonFKMRS14}). If $\Pi$ is the empty assignment, we define $\Pi(\epsilon)(0)$ as the initial state $s_0$ of $\mathcal{G}$.
Note that the games are local to each path but based on all outer paths.:
In a formula of the from $\forall \pi_1. \llangle A \rrangle \pi_2. \varphi$ the agents in $A$ know the already fixed, \emph{full} trace $\pi_1$ but behave as a strategy w.r.t. $\pi_2$. 
We write $\mathcal{G} \models \varphi$ whenever $\emptyset \models_\mathcal{G} \varphi$ where $\emptyset$ is the empty path assignment.

\begin{proposition}\label{prop:subsume}
	\HyperATLS~subsumes \HyperCTLS (and thus \HyperLTL) and \ATLS (see \refFig{expr}).
\end{proposition}

We sometimes consider \HyperATLS~formulas with extend path quantification:
We write $\llangle A \rrangle_\mathcal{G} \, \pi. \varphi$ to indicate that the path $\pi$ is the result of the game played in $\mathcal{G}$. 
We can thus refer to different structures in the same formula. 
For example, $\forall_{\mathcal{G}_1}\,\pi_1.~\llangle A \rrangle_{\mathcal{G}_2}\,\pi_2.~\ltlg (o_{\pi_1} \leftrightarrow o_{\pi_2})$ states that for each path $\pi_1$ in $\mathcal{G}_1$ the agents in $A$ have a strategy in $\mathcal{G}_2$ that produces the same outputs as on $\pi_1$. 
As for \HyperLTL, extended quantification reduces to the standard semantics \cite[\S 5.4]{DBLP:phd/dnb/Rabe16}.

\subparagraph{Parallel-Composition}
We extend \HyperATLS~with a \emph{syntactic} construct that allows multiple traces to be resolved in a single bigger game, where individual copies of the system progress in parallel.
Consider the following modification to the \HyperATLS~syntax where $k \geq 1$:
$$ \varphi := [\llangle A_1 \rrangle \pi_1.~\cdots~\llangle A_k \rrangle \pi_k.]~\varphi \mid a_\pi \mid \neg \varphi \mid \varphi \land \varphi \mid \varphi~\mathcal{U}~\varphi \mid \ltlnext \varphi$$
When surrounding strategy quantifiers by $[\cdot]$ the resulting traces are the outcome of a game played on a bigger, parallel game of the structure. 
This way, the agents in each copy can base their decisions not only on the current state of their copy but on the combined state of all $k$ copies (which allows for a coordinated behaviour among the copies).
For a player $\agent$, and a CGS $\mathcal{G} = (S, s_{0}, \agents, \mathscr{M}, \delta, \atomic, L)$, a $k$-strategy for $\agent$ is a function $f_\agent : (S^k)^+ \to \mathscr{M}$. 
The strategy can thus base its decision on a finite history of tuples of paths (in each separate copy).
For a system $\mathcal{G}$,  sets of $k$-strategies strategies $F_{A_1}, \cdots, F_{A_k}$ and states $s_1, \cdots, s_k$, we define $\mathit{out}(\mathcal{G}, (s_1, \cdots, s_k), F_{A_1}, \cdots, F_{A_k})$ as all plays $u \in (S^k)^\omega$ such that \textbf{(1)} $u(0) = (s_1, \cdots, s_k)$ and \textbf{(2)} for every $i \in \mathbb{N}$ there exist move vectors $\sigma_1, \cdots, \sigma_k$ such that $u(i+1) = \left(\delta(t_1, \sigma_1), \cdots, \delta(t_k, \sigma_k)\right)$ where $u(i) = (t_1, \cdots, t_k)$ and for every $j \in \{1, \cdots, k\}$, agent $\agent \in A_j$ and strategy $f_\agent \in F_{A_j}$ , $\sigma_j(\agent) = f_\agent(u[0,i])$.
Agents can thus control the individual progress of their system and base their decision on the history of the other quantifiers. 
Note that in case where $k = 1$ this is identical to the construction seen above. 
For simplicity we gave the semantics for a CGS (i.e., a MSCGS without stages), it can be generalized easily.
We can now extend our semantics by
{\small\begin{align*}
	\Pi &\models_\mathcal{G} [\llangle A_1 \rrangle \pi_1.~\cdots~\llangle A_k \rrangle \pi_k.]~\varphi \; \text{iff  } \\
	&\exists F_{A_1}, \cdots, F_{A_k}: \forall (t_1, \cdots, t_k) \in \mathit{out}(\mathcal{G}, (\Pi(\epsilon)(0), \cdots, \Pi(\epsilon)(0)), F_{A_1}, \cdots, F_{A_k}): \Pi[\pi_i \mapsto t_i]_{i=1}^{k} \models_\mathcal{G} \varphi
\end{align*}}%
Note that $[\llangle A \rrangle \pi.]~\varphi$ is equivalent to $\llangle A \rrangle \pi. \varphi$. This does, of course, not hold once we consider multiple strategy quantifiers grouped together by $[\cdot]$.

\subparagraph{Comparison with $\forall\exists$-HyperLTL model checking \cite{DBLP:conf/cav/CoenenFST19}}
To give some more intuition for to the self-composition, we can compare our syntactic constructs for self composition with the model checking algorithm introduced in \cite{DBLP:conf/cav/CoenenFST19}.
The idea of the method from \cite{DBLP:conf/cav/CoenenFST19} is to check a $\forall\exists$-formula by viewing the existential quantifier as a player who has to decide on a next state (in her copy) by reacting to the moves of the universal quantifier.  
If such a strategy exists, the $\forall\exists$-formula  also holds, whereas the absence of a strategy does, in general, not imply that the formula does not hold (as the strategy bases its decision on finite plays whereas the existential path is chosen with the universally quantified path already fixed).
This game based view of the existential player can be natively expressed in \HyperATLS:
While the \HyperATLS-formula $\forall \pi_1. \exists \pi_2. \varphi$ is equivalent to the same \HyperLTL-formula (i.e., the existential trace $\pi_2$ is chosen knowing the entire trace $\pi_1$), model checking of the formula $[\forall \pi_1.\exists \pi_2]. \varphi$ corresponds to the strategy search for the existential player that is only based on finite prefixes of $\pi_1$ (which directly corresponds to \cite{DBLP:conf/cav/CoenenFST19}). 
We can actually show that if any MSCGS $\mathcal{G}$ satisfies $[\forall \pi_1.\llangle A \rrangle \pi_2.] \varphi$ then it also satisfies $\forall \pi_1.~\llangle A \rrangle \pi_2. \varphi$ (see the appendix). 
This gives a more general proof of the soundness of \cite{DBLP:conf/cav/CoenenFST19}.
As our prototype implementation supports $[ \llangle A_1 \rrangle \pi_1.\llangle A_2 \rrangle \pi_2.]$-formulas, our tool subsumes the algorithm in \cite{DBLP:conf/cav/CoenenFST19} (see \refSection{proto}).

\section{Examples of Strategic Hyperproperties}\label{sec:ex}

After having introduced the formal semantics of \HyperATLS~we now demonstrate how the strategic quantification can be useful for expressing hyperproperties.
We organize our example in two categories. 
We begin with examples from information-flow control and highlight the correspondence with existing properties and security paradigms. 
Afterwards, we show how the strategic hyperproperties are naturally well suited to express asynchronous hyperproperties.

\subsection{Strategic Information-Flow Control}\label{sec:examples1}

We focus our examples on game structures that result from a reactive system. 
Let $H$ (resp.~$L$) be the set of atomic propositions forming the high-security (resp.~low-security) inputs of a system (we assume $H \cap L = \emptyset$).
The game structure then comprises $3$-players $\agent_N, \agent_H, \agent_L$, responsible for resolving non-determinism and selecting high-security and low-security inputs.
In particular, the move from $\agent_H$ (resp.~$\agent_L$) determines the values of the propositions in $H$ (resp.~$L$) in the next step. 
We call a system \emph{input-total}, if in each state, $\agent_H$ and $\agent_L$ can choose all possible valuations for the input propositions.

\subparagraph{Strategic Non-Interference}
In the introduction, we already saw that $\mathit{GNI}$ \cite{DBLP:conf/sp/McCullough88} is (in some cases) a too relaxed notion of security as it can base the existence of a witness trace on knowledge of the entire input-sequence.
Note that $\mathit{GNI}$ can be extended to also allow for input from a low-security source that may affect the output.
The \HyperATLS~formula $\mathit{stratNI}$ (given in the introduction) instead postulates a \emph{strategy} for $\agent_N$ that incrementally constructs a path that ``disproves`` information leakage.
We can show that $\mathit{stratNI}$ implies $\mathit{GNI}$. Loosely speaking, whenever there is a strategy for the non-determinism based on the finite history of inputs, there also exists a path when given the full history of inputs (as in $\mathit{GNI}$).

\begin{lemma}\label{lem:stratNIandGNI}
	For any system $\mathcal{G}$ that is input-total, we have that if $\mathcal{G} \models \mathit{stratNI}$  then $\mathcal{G} \models \mathit{GNI}$.
\end{lemma}

\subparagraph{Simulation-based Non-Interference}
Other attempts to non-interference are based on the existence of a bisimulation (or simulation) \cite{DBLP:conf/csfw/SabelfeldS00, DBLP:conf/ershov/Sabelfeld03,DBLP:conf/esorics/MantelS10}. 
While trace-based notions of non-interference (such as $\mathit{GNI}$) only require the existence of a path that witnesses the absence of a leak, simulation based properties require a lock-step relation in which this holds.
Given a system $\mathcal{G}$ with initial states $s_0$. For states $s, s'$ and evaluations $i_L \in 2^L$ and $i_H \in 2^H$, we write $s \Rightarrow^{i_L}_{i_H} s'$ if $L(s') \cap L = i_L$ and  $L(s') \cap H = i_H$ and $s'$ is a possible successor of $s$.
A \emph{security simulation} is a relation $R$ on the states of $S$ such that whenever $s R t$, we have \textbf{(1)} $s$ and $t$ agree on the output propositions, and \textbf{(2)} for any $i_L \in 2^L$ and $i_H, i'_H \in 2^H$ if $s \Rightarrow^{i_L}_{i_H} s'$ then there exists a $t'$ with $t \Rightarrow^{i_L}_{i'_H} t'$ and $s' R t'$.
Note that this is not equivalent to the fact that $\Rightarrow$ is a simulation in the standard sense \cite{DBLP:books/sp/Milner80}.
While a standard simulation relation is always reflexive, reflexivity of security simulations guarantees the security of the system \cite{DBLP:conf/ershov/Sabelfeld03,DBLP:conf/csfw/SabelfeldS00}. 
A system is thus called \emph{simulation secure} if there exists a security simulation $R$ with $s_0 R s_0$.
It is easy to see that every input-total system that is \emph{simulation secure} already satisfies $\mathit{GNI}$.
The converse does, in general, not hold.
We can show that \HyperATLS~can express \emph{simulation security} by using the parallel-composition of quantifiers. 
$$
[\forall_\mathcal{G} \, \pi_1.~\llangle \{\agent_N\} \rrangle_{\mathcal{G}_\mathit{shift}} \, \pi_2.]~\ltlg (\bigwedge_{a \in L } a_{\pi_1} \leftrightarrow \bigcirc a_{\pi_2}) \to \ltlg (\bigwedge_{a \in O} a_{\pi_1} \leftrightarrow \bigcirc a_{\pi_2})
$$
Here we consider \HyperATLS~with extended quantifier, where we annotate each quantifier with the game structure it is resolved on.
$\mathcal{G}_\mathit{shift}$ is the structure $\mathcal{G}$ where we added a dummy initial state, that shifts the behaviour of the system by one position, which is again corrected via the next operator in the \LTL~formula.
This allows the strategy for $\agent_N$ in the second copy to base its decision on already fixed step in the first copy, i.e., it corresponds to a strategy with a fixed lookahead of $1$ step. 
We can show:

\begin{lemma}\label{lem:simSec}
	A input-total system $\mathcal{G}$ is \emph{simulation secure} if and only if it satisfies $\mathit{simNI}$.
\end{lemma}

\subparagraph{Non-Deducibility of Strategies}
Lastly, we consider the notion of \emph{non-deducibility of strategies} ($\mathit{NDS}$) \cite{DBLP:conf/sp/WittboldJ90}.
$\mathit{NDS}$ requires not only that each output is compatible with each sequence of inputs but also with each input-\emph{strategy}. 
This becomes important as a high-security input player who can observe the internal state of a system might be able to leak information deliberately. 
As a motivating example consider the following (first introduced in \cite{DBLP:conf/sp/WittboldJ90}): 
Suppose we have a system that reads a binary input \myvar{h} from a high-security input and outputs \myvar{o}. 
The system maintains a bit $b$ of information in its state, initially chosen non-deterministically.
In each step, the system reads the input \myvar{h}, outputs $\text{\myvar{h}}  \oplus b$ (where $\oplus$ is the xor-operation), non-deterministically chooses a new value for $b$ and then repeats. 
As $\oplus$ essentially encodes a onetime pad it is not hard to see, that this system is secure from a purely trace-based point of view (as expressed in e.g.~$\mathit{GNI}$): 
Any possible combination of input and output can achieved when resolving the non-deterministic choice of $b$ appropriately. 
If the high-input player is, however, able to observe the system (in the context of \cite{DBLP:conf/sp/WittboldJ90} the system shares the internal bit on a private channel), she can communicate arbitrary sequence of bits to the low security environment.
Whenever she wants to send bit $c$, she inputs $c \oplus b$ where $b$ is the internal bit she has access to (note that $(c \oplus b) \oplus b = c$).
For such system system we therefore do no want to specify that every possible output sequence is compatible with all possible inputs, but instead compatible with all possible input-\emph{strategies} (based on the state of the system). 
Phrased differently, there should \emph{not} be a output sequence such that a strategy can reliably \emph{avoid} this output.
We can express $\mathit{NDS}$ in \HyperATLS~as follows:
$$
\neg \left(\exists \pi_1.~\llangle \agent_H \rrangle \pi_2.~\ltlg (\bigwedge_{a \in L} a_{\pi_1} \leftrightarrow a_{\pi_2}) \to \lozenge (\bigvee_{a \in O} a_{\pi_1} \not\leftrightarrow a_{\pi_2})\right)
$$
This formula states that there does not exist a trace $\pi_1$ such that $\agent_H$ has a strategy to avoid the output of $\pi_1$ (provided with the same low-security inputs). 
$\mathit{NDS}$ is a stronger requirement than $\mathit{GNI}$, as shown by the following Lemma:
\begin{lemma}\label{lem:NDS}
	For any system $\mathcal{G}$, if $\mathcal{G} \models \mathit{NDS}$ then $\mathcal{G} \models \mathit{GNI}$.
\end{lemma}

\subsection{Asynchronous Hyperproperties}

Reasoning about the strategic behaviour of agents is particularly useful when reasoning about asynchronous hyperproperties, as each asynchronous execution can be considered the result of the decision of an asynchronous scheduler.
We call a player an asynchronous scheduler if it can decide whether the system progresses (as decided by the other agents) or stutters.
Note that this differs from the asynchronous turn-based games defined in \cite{DBLP:journals/jacm/AlurHK02}. In our setting, the scheduler does not control which of the player controls the next move but rather decides if the system as whole progresses or stutters. 
In cases where the system does not already include such an  asynchronous scheduler (if we e.g.~use a Kripke structure interpreted as a $1$-player CGS), we can include a scheduler via a simple system transformation: 

\begin{definition}\label{def:stut}
	Given a MSCGS $\mathcal{G} = (Q, q_0, \agents, \mathscr{M}, \delta, d, \atomic, L)$ and a fresh agent \emph{sched} not already included in the set of agents of $\agents$.
	We define the stutter version of $\mathcal{G}$, denoted $\mathcal{G}_\mathit{stut}$, by $\mathcal{G}_\mathit{stut} = (Q \times \{0, 1\}, (q_0, 0), \agents \uplus \{\mathit{sched}\}, \mathscr{M} \times \{0, 1\}, \delta', d', \atomic \uplus \{\mathit{stut}\}, L')$ where 
	$$
	\delta'((s, b), \sigma) = \begin{cases}
		\begin{aligned}
			&(\delta(s,\mathit{proj}_1 \circ \sigma_{\mid \agents}), 0) \quad &&\text{if } (\mathit{proj}_2 \circ \sigma)(\mathit{sched}) = 0\\
			&(s, 1) \quad &&\text{if } (\mathit{proj}_2 \circ \sigma)(\mathit{sched}) = 1
		\end{aligned}
	\end{cases}
	$$
	$L'((s, 0)) = L(s)$ and $L'(s, 1) = L(s) \cup \{\mathit{stut} \}$. Finally $d'(\agent) = d(\agent)$ for $\agent \in \agents$ and $d'(\mathit{sched}) = m + 1$ where $m$ is the maximal element in the codomain of $d$. 
\end{definition}

Here $\mathit{proj}_i$ is the projection of the $i$th element in a tuple. 
$\mathcal{G}_\mathit{stut}$ thus progresses as $\mathcal{G}$ with the exception of the additional scheduling player. In each step, the $\{0,1\}$-decision of $\mathit{sched}$, which can be based on the decision by the other agents (as $\mathit{sched}$ is in the last stage of the game), decides if the system progresses or remains in its current state. 
The extended state-space $Q \times \{0, 1\}$ is used to keep track of the stuttering which becomes visible via the new atomic proposition $\mathit{stut}$. 
Our construction will be particularly useful when comparing our logic to \AHLTL~\cite{hyperaltl}.

\subparagraph{Observational Determinism}
As an example we consider observational-determinism which states that the output along all traces is identical. In \HyperLTL: $\mathit{OD} := \forall \pi_1. \forall \pi_2.~\ltlg (\bigwedge_{a \in O} a_{\pi_1} \leftrightarrow a_{\pi_2})$.
The example in \refFig{exProg} does not satisfy this property, as the output is changed at different timepoints. 
If we consider any system as a multi-agent system including the scheduler $\mathit{sched}$,  we can use \HyperATLS~to natrually express an asynchronous version of OD via:
$$\mathit{OD}_{\mathit{asynch}} := [\llangle \{\mathit{sched}\} \rrangle \pi_1. \llangle \{\mathit{sched}\} \rrangle \pi_2. ] \, \mathit{fair}_{\pi_1} \land \mathit{fair}_{\pi_2} \land \ltlg (\bigwedge_{a \in O} a_{\pi_1} \leftrightarrow a_{\pi_2})$$
where $\mathit{fair}_{\pi_i} := \ltlg \lozenge \neg \mathit{stut}_{\pi_i}$, asserts that the system may not be stuttered forever.
Note that we encapsulated the quantifiers by $[\cdot]$ thus resolving the games in parallel. 
The schedulers for both copies of the system can thus observe the current state of the other copy. 
The example from \refFig{exProg}, after a transformation via \refDef{stut}, satisfies this formula, as the output can be aligned by the scheduling player.

\subparagraph{One-Sided Stuttering}
By resolving the stuttered traces incrementally (i.e., omitting the $[\cdot]$-brackets) we can also express \emph{one-sided stuttering}, i.e., allow only the second copy to be stuttered. 
As an example assume $P^o$ is a program written in the high-level programming language and $P^a$ the complied program into binary code. 
Let $S^\mathit{o}$ and $S^\mathit{a}$ be the state systems of both programs.
Using \HyperATLS~we can now verify that the compiler did not leak information, i.e., the assembly code does provide the same outputs as the original code.
As the compiler breaks each program statement into multiple assembly instructions, we can not require the steps to match in a synchronous manner. Instead the system $S^o$ should be allowed to stutter for the assembly program to catch up.
We can express this as follows:
{$$\forall_{S^\mathit{a}} \, \pi_1. \; \llangle \{\mathit{sched}\} \rrangle_{S^\mathit{o}_\mathit{stut}} \,\pi_2. \mathit{fair}_{\pi_2} \land \ltlg (\bigwedge_{a \in O} a_{\pi_1} \leftrightarrow a_{\pi_2})$$}%
I.e., for every execution of the assembly code we can stutter the program such that the observations align.
Here, $S^o_\mathit{stut}$ denotes the modified version obtained by introducing an explicit scheduler (\refDef{stut}). 
Note that we again use the extend path quantifier by annotating a quantifier with the system, thereby effectively comparing both systems with respect to a hyperproperty.

\section{HyperATL* Model Checking}\label{sec:mc}

In this section we present an automata-based algorithm for \HyperATLS~model checking. 
The crucial insight in our algorithm is how to deal with the strategic quantification.
Lets briefly recall \ATLS~model checking \cite{DBLP:journals/jacm/AlurHK02}:
In \ATLS, checking if $\llangle A \rrangle \varphi$ holds in some state $s$, can be reduced to the non-emptiness check of the intersection of two tree automata. 
One accepting all possible trees that can be achieved via strategy for players in $A$, and one accepting all trees whose paths satisfy the path formula $\varphi$ \cite{DBLP:journals/jacm/AlurHK02}.
In our hyperlogic this is not possible.
When checking $\llangle A \rrangle \pi. \varphi$ we can not construct an automaton accepting all trees that satisfy $\varphi$, as the satisfaction of $\varphi$ depends on the the paths assigned to the outer path-quantifiers (that are not yet fixed).
Instead, we construct an automaton that accepts all \emph{path assignments} for the outer quantifiers such that there exists a winning strategy for the agents in $A$.
We show that alternating automata are well suited to keep track of all path assignments for which a strategy exists as they allow us to encode the strategic behaviour of $\mathcal{G}$ within the transition function of the automaton.

\begin{figure}
	\begin{tabular}{|c|c|}
		\hline
		$\stackanchor{$a_{\pi_i}$}{\color{red}$\neg a_{\pi_i}$}$ & $\begin{aligned}
			&\mathcal{A}_\varphi = (\{q_\mathit{init}\}, q_\mathit{init}, \Sigma_\varphi, \rho, \mathbf{0})\\
			&\rho(q_\mathit{init}, [s_1, \cdots, s_n]) = \begin{cases}
				\begin{aligned}
					&&\stackanchor{$\top$}{\color{red}$\bot$} \quad &&&\text{if } a \in L(s_i)\\
					&&\stackanchor{$\bot$}{\color{red}$\top$} \quad &&&\text{if } a \not\in L(s_i)
				\end{aligned}
			\end{cases}
		\end{aligned}$ \\\hline
		$\varphi_1 \stackanchor{$\lor$}{\color{red}$\land$} \varphi_2$ & $\begin{aligned}
			&\mathcal{A}_\varphi = (Q_1 \cup Q_2 \cup \{q_\mathit{init}\}, q_\mathit{init}, \Sigma_\varphi, \rho, c_1 \uplus c_2 \uplus [q_\mathit{init} \mapsto 0])\\
			&\rho(q, [s_1, \cdots, s_n]) = \begin{cases}
				\begin{aligned}
					&&\rho_1(q_{0, 1}, [s_1, \cdots, s_n])\, \stackanchor{$\lor$}{\color{red}$\land$}\, \rho_2(q_{0, 2}, [s_1, \cdots, s_n]) \quad &&&\text{if } q = q_\mathit{init}\\
					&&\rho_i(q, [s_1, \cdots, s_n]) \quad &&&\text{if } q \in Q_i
				\end{aligned}
			\end{cases}
		\end{aligned}$  \\\hline
		$\ltlnext \varphi_1$ & $\begin{aligned}
			&\mathcal{A}_\varphi = (Q_1 \cup\{q_\mathit{init}\}, q_\mathit{init}, \Sigma_\varphi, \rho, c_1 \uplus [q_\mathit{init} \mapsto 0])\\
			&\rho(q, [s_1, \cdots, s_n]) = \begin{cases}
				\begin{aligned}
					&&q_{0, 1} \quad &&&\text{if } q = q_\mathit{init}\\
					&&\rho_1(q, [s_1, \cdots, s_n]) \quad &&&\text{if } q  \in Q_1
				\end{aligned}
			\end{cases}
		\end{aligned}$  \\\hline
		$\varphi_1 \stackanchor{$\mathcal{U}$}{\color{red}$\mathcal{R}$} \varphi_2$ & $\begin{aligned}
			&\mathcal{A}_\varphi = (Q_1 \cup Q_2 \cup \{q_\mathit{init}\}, q_\mathit{init}, \Sigma_\varphi, \rho, c_1 \uplus c_2 \uplus [q_\mathit{init} \mapsto \stackanchor{$1$}{\color{red}$0$}])\\
			&\rho(q, [s_1, \cdots, s_n]) = \begin{cases}
				\begin{aligned}
					&&\rho_2(q_{0, 2}, [s_1, \cdots, s_n]) \stackanchor{$\lor$}{\color{red}$\land$} \big(\rho_1(q_{0, 1}, [s_1, \cdots, s_n]) \stackanchor{$\land$}{\color{red}$\lor$} q_\mathit{init}\big) \quad &&&\text{if } q = q_\mathit{init}\\
					&&\rho_i(q, [s_1, \cdots, s_n]) \quad &&&\text{if } q \in Q_i
				\end{aligned}
			\end{cases}
		\end{aligned}$\\\hline
	\end{tabular}
	\caption{APA construction for \LTL~temporal operators. $\mathcal{A}_{\varphi_i} = (Q_i, q_{0, i}, \Sigma_{\varphi_i}, \rho_i, c_i)$ is the inductively constructed automaton for $\varphi_i$.}\label{fig:ltltran}
\end{figure}

\subsection{Automata-Theoretic Model Checking}

We assume that the formula $\varphi$ to be checked is given in \emph{negation normal form}, i.e., negations only occur directly in front of atomic propositions or in front of a strategy quantifier. By including conjunction ($\land$) and release ($\mathcal{R}$) every formula can be translated into a negation normal form of linear size.
We, furthermore, assume that if $\neg \llangle A \rrangle \pi. \varphi$ occurs in the formula we have $A \neq \emptyset$.
Note that in this case where $A = \emptyset$ we can push the negation in as $\neg \forall \pi. \varphi \equiv \exists \pi. \neg \varphi$.
For infinite words  $t_1, \cdots, t_n \in \Sigma^\omega$ we define $\mathit{zip}(t_1, \cdots, t_n) \in (\Sigma^n)^\omega$ as the word obtained by combining the traces pointwise, i.e.,  $\mathit{zip}(t_1, \cdots, t_n)(i) := (t_1(i), \cdots, t_n(i))$. 
Our algorithm now progresses in a bottom-up manner.
Assume that some subformula $\varphi$ occurs under quantifiers binding path variables $\pi_1, \cdots, \pi_n$.
We say that an automaton $\mathcal{A}$ over $S^n$ is $\mathcal{G}$-equivalent to $\varphi$, if for any paths $t_1, \cdots, t_n$ it holds that $[\pi_i \mapsto t_i]_{i=1}^n \models_\mathcal{G} \varphi$ if any only if $\mathit{zip}(t_1, \cdots, t_n) \in \mathcal{L}(\mathcal{A})$.
$\mathcal{G}$-equivalence thus means that an automaton accepts the zipping of traces exactly if the trace assignment constructed from them satisfies the formula.
By induction on the structure of the formula we construct an automaton that is $\mathcal{G}$-equivalent to each sub formula.

For the standard boolean combinators and \LTL~temporal operators our construction follows the typical translation from \LTL~to alternating automata \cite{DBLP:conf/lics/MullerSS88,DBLP:books/sp/cstoday95/Vardi95} given in \refFig{ltltran}. 
The interesting case is now the elimination of a strategy quantifier of the from $\varphi = \llangle A \rrangle \pi. \psi$.
Given an inductively constructed APA $\mathcal{A}_\psi$ over $\Sigma_\psi = S^{n+1}$. We aim for an automaton $\mathcal{A}_\varphi$ over $\Sigma_\varphi = S^{n}$.
The automata should accept all traces $t$ over $S^{n}$ such that there exist a strategy for agents in $A$ such that all traces compatible with this strategy $t'$ when added to $t$ (the trace $t \times t' \in (S^{n+1})^\omega$) is accepted by $\mathcal{A}_\psi$.
Let $\mathcal{G} = (S, s_0, \agents, \mathscr{M}, \delta, d, \atomic, L)$ be the given MSCGS.
We distinguish between the cases where $A = \agents$ (i.e., existential quantification) and $A \neq \agents$.

\subparagraph{Existential Quantification}
We first consider the case where $A = \agents$, i.e., $\varphi = \exists \pi. \psi$. Model checking can be done similar to \cite{DBLP:conf/cav/FinkbeinerRS15}. Let $\mathcal{A}_\psi = (Q, q_0, \Sigma_\psi, \lambda : Q \times \Sigma_\psi \to 2^Q, c)$ be the inductively constructed automaton, translated into a non-deterministic automaton of exponential size via \refTheo{alt1}.
We then construct $\mathcal{A}_\varphi := (S \times Q \cup \{q_{\mathit{init}}\}, q_{\mathit{init}}, \Sigma_\varphi, \rho, c')$ where $c'(s, q) = c(q)$ and $\rho$ is defined via
\begin{align*}
	\rho(q_{\mathit{init}}, [s_1, \cdots, s_n]) &= \{(s', q') \mid q' \in \lambda(q, [s_1, \cdots, s_n, s_n^\circ]) \land \exists \sigma : \agents \to \mathscr{M}. \delta(s_n^\circ, \sigma) = s'\} \\
	\rho((s, q), [s_1, \cdots, s_n]) &= \{(s', q') \mid q' \in \lambda(q, [s_1, \cdots, s_n, s]) \land \exists \sigma : \agents \to \mathscr{M}. \delta(s, \sigma) = s'\}
\end{align*}
where we define $s_n^\circ = s_n$ if $n \geq 1$ and $s_n^\circ = s_0$ otherwise.
Note that $\mathcal{A}_\varphi$ is again a non-deterministic automaton.
Every accepting run of $\mathcal{A}_\varphi$ on $\mathit{zip}(t_1, \cdots, t_n)$ now corresponds to a path $t$ in $\mathcal{G}$ such that $\mathcal{A}_\psi$ accepts $\mathit{zip}(t_1, \cdots, t_n, t)$.

\subparagraph{(Complex) Strategic Quantification }
We now consider the case where $A \neq \agents$. Our automaton must encode the strategic behaviour of the agents.
We achieve this, by encoding the strategic play in the game structure within the transition function of an automaton.
Let $\mathcal{A}^\mathit{det}_\psi = (Q, q_0, \Sigma_\psi, \lambda : Q \times \Sigma_\psi \to Q, c)$ be a \emph{deterministic} automaton obtained from the inductively constructed $\mathcal{A}_\psi$ via \refTheo{alt1}.
Note that $\mathcal{A}^\mathit{det}_\psi$ is, in the worst case, of double exponential size (in the size of $\mathcal{A}_\psi$).
To encode the strategic behaviour in $\mathcal{G}$ we use the alternation available in an automaton by disjunctively choosing moves for controlled players in $A$, followed by a conjunctive treatment of all adversarial player. 
The stages of a game, naturally correspond to the order of the move selection.
Define the set $A_i := A \cap d^{-1}(i)$ and $\overline{A}_i := (\agents \setminus A) \cap d^{-1}(i)$ and let $m$ be the maximal element in the codomain of $d$.
The choice of each agent in $A$ followed by those not in $A$ can be encoded into a boolean formula. 
We define $\mathcal{A}_\varphi := (S \times Q \cup \{q_\mathit{init}\}, q_\mathit{init}, \Sigma_\varphi, \rho, c')$ where $\rho$ is defined by
{\small\begin{align*}
	\rho(q_\mathit{init}, [s_1, \cdots, s_n]) &= \bigvee\limits_{\sigma_1 : A_1 \to \mathscr{M}} \bigwedge\limits_{\sigma'_1 : \overline{A}_1 \to \mathscr{M}} \cdots \bigvee\limits_{\sigma_m : A_m \to \mathscr{M}} \bigwedge\limits_{\sigma'_m : \overline{A}_m \to \mathscr{M}}   
	\big(\delta(s^\circ_n, \sum_{i=1}^{m} \sigma_i  + \sigma'_1), \lambda(q_0, [s_1, \cdots, s_n, s^\circ_n])   \big)\\
	\rho\left((s, q), [s_1, \cdots, s_n]\right) &= \bigvee\limits_{\sigma_1 : A_1 \to \mathscr{M}} \bigwedge\limits_{\sigma'_1 : \overline{A}_1 \to \mathscr{M}} \cdots \bigvee\limits_{\sigma_m : A_m \to \mathscr{M}} \bigwedge\limits_{\sigma'_m : \overline{A}_m \to \mathscr{M}}   
	\big(\delta(s, \sum_{i=1}^{m} \sigma_i  + \sigma'_1), \lambda(q, [s_1, \cdots, s_n, s])   \big)
\end{align*}}%
and $c'(s, q) = c(q)$ (We can define $c'(q_\mathit{init})$ arbitrarily). 
In case $n = 0$, we again define $s^\circ_n$ as the initial state $s_0$, otherwise $s^\circ_n = s_n$.
Note that in the case where the MSCGS is a CGS, i.e., $d = \mathbf{0}$ the transition function has the from $\lor \land$, where the choices in $A$ are considered disjunctively and the choices by all other agents conjunctively.  
Our construction can be extended to handle the self composition $[\llangle A_1 \rrangle \pi_1.~\cdots~\llangle A_k \rrangle \pi_k]$ (see the appendix for details). 

\subparagraph{Negated Quantification}
We extend our construction to handle negation outside of quantifiers, i.e., a formula $\varphi = \neg \llangle A \rrangle \pi. \psi$ via $\mathcal{A}_\varphi := \overline{\mathcal{A}_{\llangle A \rrangle \pi. \psi}}$ by using \refTheo{altneg}.

\begin{proposition}\label{prop:MCequiv}
	$\mathcal{A}_\varphi$ is $\mathcal{G}$-equivalent to $\varphi$. 
\end{proposition}

By following our inductive construction we obtain an automaton over the singleton alphabet (empty state sequences) that is non-empty iff the model satisfies the formula.
Emptiness of an alternating parity automaton can then be checked in polynomial size (assuming a fixed number of colours) \cite{DBLP:conf/stoc/KupfermanV98}.

We can observe a gap in complexity of algorithm between simple and complex quantification. The former case requires a translation of an alternating automata to a non-deterministic one whereas the latter requires a full determinisation.
To capture the complexity of our algorithm we define $\mathcal{T}_c(k, n)$ as a tower of $k$ exponents in $n$, i.e., $\mathcal{T}_c(0, n) = n^c$ and $\mathcal{T}_c(k+1, n) = c^{\mathcal{T}_c(k, n)}$.
For $k \geq 0$ we define k-\EXPSPACE~as the class of languages recognised by a deterministic Turing machine (\cite{DBLP:journals/jcss/Savitch70}) with space $\mathcal{T}_c(k, \mathcal{O}(n))$ for some $c$ (and similarly for time).
We define $(-1)$-\EXPSPACE~as \NLOGSPACE. Note that $0$-\EXPSPACE~=~\PSPACE.

\begin{theorem}\label{theo:upperBound}
	Model checking of a \HyperATLS~formula with $k$ complex and $l$ simple quantifiers is in $(2k + l)$-\EXPTIME. If $l \geq 1$ and the formula is linear it is also in $(2k+l-1)$-\EXPSPACE~(both in size of the formula). 
\end{theorem}

The fact that we can derive a better upper bound when $l > 0$ follows from the fact that we can determine the emptiness of a non-deterministic automaton in \NLOGSPACE~\cite{DBLP:journals/iandc/VardiW94} and for an alternating automaton only in polynomial time (for a fixed number of colours) \cite{DBLP:conf/stoc/KupfermanV98}. 
Note that for the syntactic fragment of \HyperCTLS~our algorithm matches the algorithm in \cite{DBLP:conf/cav/FinkbeinerRS15}.

\section{Lower Bounds For HyperATL* Model Checking}\label{sec:lb}

\refTheo{upperBound} gives us an upper bound on the model checking problem for \HyperATLS.
We can show the following lower bound

\begin{theorem}\label{theo:lb}
	Model checking of a linear \HyperATLS~formula with $k$ complex and $l$ simple quantifiers is $(2k + l - 1)$-\EXPSPACE-hard in the size of the formula, provided $l \geq 1$.
\end{theorem}

The proof of \refTheo{lb} proceeds by encoding space-bounded Turing machines into \HyperATLS.
We show that (complex) strategic quantification can be used to encode a incremental counter that grows by \emph{two exponents} with each quantifier, opposed to the increment by a single exponent for simple quantification \cite{DBLP:phd/dnb/Rabe16}.
The curial idea that allows us to encode a doubly-exponential counter is the fact that we can use a player within the game to check the correctness. While the first player can generate a counter, the second player tries to show that the counter is wrong. The only winning strategy for the former player is then to output a correct counter that holds up against all scrutiny by the latter player. 
As the construction of the counter is rather complex we refer the interested reader to a detailed proof in the appendix.

Note that \refTheo{lb} is conditioned on $l \geq 1$. This gives an interesting complexity landscape:
In cases where $l \geq 1$, model checking is $(2k+l-1)$-\EXPSPACE-complete (irrespective of $k$). 
If $l = 0$ we get an upper bound of $2k$-\EXPTIME~(\refTheo{upperBound}). 
In the special case where $l=0$ and $k =1$ we get a matching lower bound from the \ATLS~model-checking problem \cite{DBLP:journals/jacm/AlurHK02} (subsuming \LTL~realizability \cite{DBLP:conf/popl/PnueliR89,DBLP:conf/icalp/PnueliR89}) and thus $2$-\EXPTIME-completeness. If $k > 1$ the best lower bound is $(2k-2)$-\EXPSPACE. The exact complexity for the case where $k > 1$ and $l = 0$ is thus still open. 

\section{HyperATL* vs.~asynchronous HyperLTL}\label{sec:rel}

We have seen that our strategic logic can naturally express asynchronous hyperproperties.
In this section we compare our logic to \AHLTL~\cite{hyperaltl}, a recent extension of \HyperLTL~specifically designed to express such asynchronous properties. 
\AHLTL~is centred around the idea of a trajectory, which, informally speaking, is the stuttering of traces in a system.
In \AHLTL~an initial trace quantifier prefix is followed by a quantification over such a trajectory.
For example, a formula of the form $\forall \pi_1. \cdots \forall \pi_n. \mathbf{E}. \varphi$ means that for all paths $\pi_1, \cdots, \pi_n$ in the system there exists some stuttering of the paths, such that $\varphi$ is satisfied. 
\AHLTL~follows a purely trace-based approach where the stuttering is fixed knowing the full paths $\pi_1, \cdots, \pi_n$. In comparison, in our logic a \emph{strategy} must decide if to stutter based on finite a prefix in the system. 
Model checking \AHLTL~is, in general, undecidable \cite{hyperaltl}.
The largest known fragment for which an algorithm is known are formulas of the form $\forall \pi_1. \cdots \forall \pi_n. \mathbf{E}. \varphi$ where $\varphi$  is an \emph{admissible formula} \cite{hyperaltl} which is a conjunction of formulas of the form $\ltlg \bigwedge_{a \in P} a_{\pi_i} \leftrightarrow a_{\pi_j}$ (where $P$ is a set of atomic propositions) and stutter-invariant formulas over a single path variable.  
We can show the following (where $\mathcal{G}_\mathit{stut}$ is the stutter transformation from \refDef{stut}):

\begin{theorem}\label{theo:altl}
	For any Kripke structure $\mathcal{G}$ and \AHLTL~formula of the form $\forall \pi_1. \cdots \forall \pi_n. \mathbf{E}. \varphi$ it holds that if $\mathcal{G}_\mathit{stut} \models [\llangle \mathit{sched} \rrangle \pi_1.\cdots\llangle \mathit{sched} \rrangle \pi_n]~\varphi \land \bigwedge_{i \in \{1, \cdots, n\}} \mathit{fair}_{\pi_i}$ \textbf{(1)} then $\mathcal{G} \models_{\texttt{AHLTL}} \forall \pi_1. \cdots \forall \pi_n. \mathbf{E}. \varphi$ \textbf{(2)}.
	If $\varphi$ is an admissible formula, \textbf{(1)} and \textbf{(2)} are equivalent.
\end{theorem}

\refTheo{altl} gives us a sound approximation of the (undecidable) \AHLTL~model checking.
Furthermore, for admissible formulas, \AHLTL~can be truthfully expressed in our logic. 
As shown in \cite{hyperaltl}, many interesting properties can be expressed using an admissible formula and can thus be (truthfully) checked in our logic.
Our framework can therefore express many interesting properties, is fully decidable and also subsumes the largest known decidable fragment of \AHLTL.

\section{Experimental Evaluation}\label{sec:proto}

While MC for the full logic is very expensive (\refTheo{lb}) and likely not viable in practice, formulas of the from $[\llangle A_1 \rrangle \pi_1.~\cdots \llangle A_n\rrangle \pi_n.] \varphi$ where $\varphi$ is quantifier free, can be checked very efficiently via a reduction to a parity game (see the appendix for details). 
Note that all alternation-free \HyperLTL~formulas, the reduction from the MC approach from \cite{DBLP:conf/cav/CoenenFST19} and the reduction in \refTheo{altl} fall in this fragment.
We implemented a prototype model checker for this fragment to demonstrate different security notions (both synchronous and asynchronous) on small example programs.
Our tool uses \texttt{rabinizer 4} \cite{DBLP:conf/cav/KretinskyMSZ18} to convert a \LTL~formula into a deterministic automaton and \texttt{pgsolver} \cite{DBLP:conf/atva/FriedmannL09} to solve parity games.
Our tool is available on GitHub at \url{https://github.com/reactive-systems/hyperatl}.

\subparagraph{Information-Flow Policies}
We have created a small benchmark of simple programs that distinguish different information-flow policies.
We checked the following properties:
\textbf{(OD)} is the standard (alternation-free) formula of observational determinism,
\textbf{(NI)} is a simple formulation of non-interference due to \cite{DBLP:conf/sp/GoguenM82a},
\textbf{(simSec)} is simulation security \cite{DBLP:conf/ershov/Sabelfeld03} as expressed in \refSection{examples1}.
Finally, \textbf{(sGNI)} is the simple game based definition of GNI resolved on the parallel-composition (as used in \cite{DBLP:conf/cav/CoenenFST19}). 
We designed small example programs that demonstrate the difference between security guarantees and present the results in \refTable{res1}.
Note that the model checking algorithm for $\forall^*\exists^*$ formulas from \cite{DBLP:conf/cav/CoenenFST19} is subsumed by our approach. As we reduce the search of a strategy for the existential player to a parity game opposed to a SMT constraint, we can handle much bigger systems (thousands of states vs.~$\leq 3$ states in \cite{DBLP:conf/cav/CoenenFST19}).

\subparagraph{Asynchronous Hyperproperties}
To showcase the expressiveness of our framework to handle asynchronous properties, we implemented the stuttering transformation from \refDef{stut}.
We evaluated our tool by checking example programs both on synchronous observational-determinism \textbf{(OD)} and asynchronous versions of OD \textbf{(OD$_{asynch}$)} and non-interference \textbf{(NI)$_{asynch}$}.  
Note that while \textbf{(OD$_{asynch}$)} can also express in the decidable fragment of \AHLTL, \textbf{(NI$_{asynch}$)} is not an admissible formula (and can not be handled in \cite{hyperaltl}). 
As non-interference only requires the outputs to align provided the inputs do, one needs to take care that the asynchronous scheduler does not ``cheat`` by deliberately missaligning inputs and thereby invalidating the premise of this implication. 
Our results are given in \refTable{res2}.
To demonstrate the state-explosion problem we tested the same program (\texttt{Q1}) with different bit-widths (programs \texttt{Q1}$_{i}$, \texttt{Q1}$_{ii}$, \texttt{Q1}$_{iii}$), causing an artificial blow-up in the number of states. 

\begin{table}
	
	\begin{subtable}{0.5\textwidth}
		\begin{tabular}{|c|c|c|c|c|}
			\hline
			& \textbf{(OD)} & \textbf{(NI)} &  \textbf{(simSec)} & \textbf{(sGNI)} \\
			\hline
			\texttt{P1} & \cmark ($15$) & \cmark ($16$) & \cmark ($16$) & \cmark ($46$) \\
			\hline
			\texttt{P2} & \xmark ($112$) & \cmark ($80$) & \cmark ($83$) & \cmark ($432$)\\
			\hline
			\texttt{P3} & \xmark ($70$) & \xmark ($44$) & \cmark ($54$) & \cmark ($112$) \\
			\hline
			\texttt{P4} & \xmark ($73$) & \xmark ($64$) & \xmark ($70$) & \cmark ($191$) \\
			\hline
		\end{tabular}
		\vspace{0.2cm}
		\subcaption{Examples for Information-flow policies}\label{tab:res1}
	\end{subtable}
	\begin{subtable}{0.5\textwidth}
		\begin{tabular}{|c|c|c|c|}
			\hline
			& \textbf{(OD)} & \textbf{(OD$_{asynch}$)} & \textbf{(NI$_{asynch}$)} \\
			\hline
			\texttt{Q1}$_{i}$ & \xmark ($112$) &\cmark ($788$) & \cmark ($812$) \\
			\hline
			\texttt{Q1}$_{ii}$ & \xmark ($281$) &\cmark ($3372$) & \cmark ($3516$) \\
			\hline
			\texttt{Q1}$_{iii}$ & \xmark ($1680$) &\cmark ($20756$) & \cmark ($24078$) \\
			\hline
			\texttt{Q2} & \xmark ($985$) &\xmark ($18141$) & \cmark ($6333$) \\
			\hline
		\end{tabular}
		\vspace{0.2cm}
		\subcaption{Examples for Asynchronous hyperproperties}\label{tab:res2}
	\end{subtable}
	
	\caption{Validity of various \HyperATLS~formulas on small benchmark programs. A \cmark (resp.~\xmark) means that the formula is satisfied (resp.~not satisfied). The time consumption is given in milliseconds. }
\end{table}

\section{Related Work}

There has been a lot of recent interest in logics for hyperproperties. Most logics are obtained by extending standard temporal or first-order/second-order logics with either path quantification or by a special equal-level predicate~\cite{DBLP:conf/stacs/Finkbeiner017}.   See \cite{DBLP:conf/lics/CoenenFHH19} for an overview.
To the best of our knowledge, none of these logics can express strategic hyperproperties.

\subparagraph{Alternating-time Temporal Epistemic Logic}
The relationship between epistemic logics and hyperlogic is interesting, as both reason about the flow of information in a system. 
As shown in \cite{DBLP:conf/fossacs/BozzelliMP15}, \HyperLTL~and \LTL$_\mathcal{K}$ (\LTL~extended with a knowledge operator \cite{DBLP:books/mit/FHMV1995}) have incomparable expressiveness.
In \texttt{HyperQPTL}, which extends \HyperLTL~with propositional quantification~\cite{DBLP:conf/cav/FinkbeinerHHT20}, the knowledge operator can be encoded by explicitly marking the knowledge positions via propositional quantification \cite[\S 7]{DBLP:phd/dnb/Rabe16}.
Alternating-time temporal logic has also been extended with knowledge operators \cite{DBLP:journals/sLogica/HoekW03a}. The resulting logic, \texttt{ATEL}, can express properties of the form ``if $\agent$ knows $\phi$, then she can enforce $\psi$ via a strategy.''
The natural extension of the logic in \cite{DBLP:journals/sLogica/HoekW03a}, which allows for arbitrary nesting of quantification and operators (i.e., an extension of \ATLS~instead of \ATL) is incomparable to \HyperATLS.

\subparagraph{Model Checking}
Decidable model checking is a crucial prerequisite for the effective use of a logic. Many of the existing (synchronous) hyperlogics admit decidable model checking, although mostly with non-elementary complexity (see \cite{DBLP:conf/vmcai/Finkbeiner21} for an overview). 
For alternating-time temporal logic (in the non-hyper realm), model checking is efficient (especially when one prohibits arbitrary nesting of temporal operators and quantifiers as in \ATL) \cite{DBLP:journals/jacm/AlurHK02,DBLP:conf/cav/AlurHMQRT98}. 
If one allows operators and quantifiers to be nested arbitrarily (\ATLS), model checking subsumes \LTL~satisfiability and realizability. This causes a jump in the model checking complexity to $2$-\EXPTIME-completeness.
As our lower bound demonstrates, the combination of strategic quantification and hyperproperties results in a logic that is algorithmically harder (for model checking) than non-strategic hyperproperties (as \HyperLTL) or non-hyper strategic logics (as \ATLS). 
The fragment of \HyperATLS~implemented in our prototype model checker subsumes alteration-free \HyperLTL~(see MCHyper \cite{DBLP:conf/cav/FinkbeinerRS15}), model checking via explicit strategies \cite{DBLP:conf/cav/CoenenFST19} and the (known) decidable fragment of \AHLTL~\cite{hyperaltl}.

\subparagraph{Asynchronous Hyperproperties}
Extending hyperlogics to express asynchronous properties has only recently started to gain momentum \cite{DBLP:journals/pacmpl/GutsfeldMO21,hyperaltl,lics2021hyper}. 
In \cite{hyperaltl} they extend \HyperLTL~with explicit trajectory quantification. 
\cite{DBLP:journals/pacmpl/GutsfeldMO21} introduced a variant of the polyadic $\mu$-calculus, $H_\mu$, able to express hyperproperties.
In \cite{lics2021hyper} they extended \HyperLTL~with new modalities that remove redundant (for example stuttering) parts of a trace.
Model checking is undecidable for all three logics.
The (known) decidable fragment of \cite{hyperaltl} can be encoded into \HyperATLS. The only known decidable classes for $H_\mu$ \cite{DBLP:journals/pacmpl/GutsfeldMO21} and $\texttt{HyperLTL}_S$ \cite{lics2021hyper} are obtained by bounding the asynchronous offset by a constant $k$, i.e., asynchronous execution may not run apart (``diverge'') for more than $k$ steps. For actual software, this is a major restriction. 

\section{Conclusion}

We have introduced \HyperATLS, a temporal logic for strategic hyperproperties. Besides the obvious benefits of simultaneously reasoning about strategic choice and information flow, \HyperATLS
provides a natural formalism to express \emph{asynchronous} hyperproperties, which has been a major challenge for previous hyperlogics.
Despite the added expressiveness, \HyperATLS~model checking remains decidable, with comparable cost to logics for synchronous hyperproperties (cf.~\refTheo{upperBound}).
\HyperATLS~is the first logic for asynchronous hyperproperties where model checking is decidable for the entire logic. Its expressivness and decidability, as well as the availability of practical model checking algorithms, make it a very promising choice for model checking tools for hyperproperties.



\bibliography{references}

\appendix

\section{Additional Material for \refSection{ex} - HyperATL*}

\newtheorem*{reimpl}{Restatement of Theorem \ref{theo:impl}}

\begin{theorem}
	For any MSCGS $\mathcal{G}$ and quantifier-free formula $\varphi$ it holds that if $\mathcal{G} \models [\forall \pi_1. \llangle A \rrangle \pi_2.] \varphi$ then $\mathcal{G} \models \forall \pi_1. \llangle A \rrangle \pi_2. \varphi$.
\end{theorem}
\begin{proof}
	To keep the notation simple we consider the case where the MSCGS has a single stage. The proof generalizes easily. 
	Assume that $\mathcal{G} \models [\forall \pi_1. \llangle A \rrangle \pi_2.] \varphi$. Let $F_A$ be the set of strategies for the agents in $A$ that is wining. 
	For each $\agent \in A$ we thus have a strategy $f_\agent : (S \times S)^+ \to \mathscr{M}$.

	We show that $\mathcal{G}  \models \forall \pi_1. \llangle A \rrangle \pi_2. \varphi$. Let $\Pi = [\pi_1 \mapsto t_1]$ be any path assignment for $\pi_1$ (which is universally quantified).
	We construct a winning strategy $f_\agent' : S^+ \to \mathscr{M}$. We define the function by
	$$f_\agent'(u) := f_\agent(u \times t_1[0, |u|])$$
	Here $\times$ denote the pointwise product of two traces.
	$f_\agent'$ thus disregards most of the already fixed path $t_1$ and simply queries $f_\agent$ based on the current prefix and the prefix of $t_1$. 
	
	Let $F_A'$ be the agent set for agents in $A$ be created in this way. It is easy to see that for each $t \in \mathit{out}(\mathcal{G}, s_0, F_A')$ the traces $u \times t_1$ (again taken pointwise) is in $\mathit{out}(\mathcal{G}, (s_0, s_0), \emptyset, F_A)$, i.e. allowed by the strategies in the self-composed version.
	$F_A'$ is thus a winning set of strategies and $\mathcal{G}  \models \forall \pi_1. \llangle A \rrangle \pi_2. \varphi$.
\end{proof}

\section{Additional Material for \refSection{ex} - Examples of Strategic Hyperproperties}

\newtheorem*{restratNIandGNI}{Restatement of Lemma \ref{lem:stratNIandGNI}}

\begin{restratNIandGNI}
	For any system $\mathcal{G}$ that is input-total, we have that if $\mathcal{G} \models \mathit{stratNI}$  then $\mathcal{G} \models \mathit{GNI}$.
\end{restratNIandGNI}
\begin{proof}
	Assume for contradiction that $\mathcal{G} \not\models \mathit{GNI}$, there thus exists traces $\pi_H$ and $\pi_{O}$ such that no trace in $\mathcal{G}$ agrees with the high inputs of $\pi_H$ and outputs of $\pi_{O}$.
	We show that $\mathcal{G} \not\models \mathit{stratNI}$. As the reference trace $\pi_1$ we choose $\pi_{O}$. 
	It is now easy to see that $\agent_N$ has no strategy to enforce $\mathit{stratNI}$, as the non-controlled input players can always choose the high inputs inputs from $\pi_H$ (as the system is input-total).
	The strategy for $\agent_N$ would thus need to construct a trace that does not exists in the system (even if we quantify existentially).
\end{proof}

\newtheorem*{reSeimSec}{Restatement of Lemma \ref{lem:simSec}}

\begin{reSeimSec}
	A input-total system $\mathcal{G}$ is \emph{simulation secure} if and only if it satisfies $\mathit{simNI}$.
\end{reSeimSec}
\begin{proof}
	We restrict ourself to a high-level proof omitting some details. The basic idea is similar to the fact that bisimulation (and simulation) can be characterised by simple turn based two player games. 
	
	\proofsubparagraph*{First Direction}
	Assume $\mathcal{G}$ is \emph{simulation secure} an let $R$ be the security simulation witnessing this.  
	We show that $\agent_N$ has a strategy in the parallel-composition that always guarantees that the states of both systems are related by $R$ (after shifting) or the low-security input on the paths leading to the states differs.
	If the low-security inputs differ, the premise in $\mathit{simNI}$ is invalid so the the formula is satisfied. 
	If the strategy can otherwise guarantee to stay in $R$ related states,  the output agrees by definition of a security simulation, so the implication is trivially true and the strategy winning.
	
	Now lets construct the wining strategy for $\agent_N$, that maintains in $R$ related states.
	Initially this is obvious as by assumption $s_0 R s_0$. 
	Note that due to the shifting, the $\agent_N$ can base its decision already knowing the move made in the other system.
	Let $s$ and $t$ be the states in both copies and $s R t$ (where we already accounted for the shifting)
	Now the $\agent_H$ and $\agent_L$ player can choose input propositions for the next state $s'$ on $\pi_1$. Similarly $\agent_H$ and $\agent_L$ choose the input for next state, called $i_L, i_H$.
	In case $\agent_N$ in the first and second copy choose different low-security inputs the statement is trivial, as the premise of the implication is false.
	If they choose the same low-security input propositions, by using $R$ there exists a state $t$ with $t \Rightarrow^{i_L}_{i_H} t'$ and $s' R t'$. $\agent_N$-strategy now selects this state. 
	The resulting strategy thus guarantees to be in $R$-related states.
	
	\proofsubparagraph*{Second Direction}
	The second direction follows in a similar way. Assume $\agent_N$ has a winning strategy. We construct the relation $R$ by defining $s R t$ iff there is a initial path of states on $\pi_1$ ending in $s$ such that the strategy for $\agent_N$ (on this initial path) allows a path ending in $t$ (and the initial paths agrees on the low inputs).
	By definition of $A$ and $\mathit{simNI}$ its is easy to see that $R$-related states agree on the outputs.
	Furthermore $R$ is a security simulation as every step in the first component must be matches in the second (as the strategy is wining).
\end{proof}

\subparagraph{Non-Deducibility of Strategies}

As a motivating example consider the following (first introduced in \cite{DBLP:conf/sp/WittboldJ90}): 
Suppose we have a system that reads a binary input \myvar{h} from a high-security input and outputs \myvar{o}. 
The system maintains a bit $b$ of information in its state, initially chosen non-deterministically.
In each step, the system reads the input \myvar{h}  and proceeds as follows: It outputs $\text{\myvar{h}}  \oplus b$ (where $\oplus$ is the xor-operation) and non-deterministically chooses a new value for $b$.
As $\oplus$ essentially encodes a onetime pad it is not hard to see, that this system is secure from a purely trace-based point of view: 
Every possible inputs sequence of \myvar{h} is compatible with any possible sequence of outputs so the system satisfies $\mathit{GNI}$.
If the high-input player is however able to observe the system, (in the context of \cite{DBLP:conf/sp/WittboldJ90} the system shares the internal bit on a private channel), she can communicate arbitrary sequence of bits to the low security environment.
Whenever she wants to send bit $c$, she inputs $c \oplus b$ where $b$ is the internal bit she has access to.
In particular the high-security inputs can be chosen such that a certain output sequence (which can be generated by the system) is impossible under that particular strategy. 
This leads to the notion of  \emph{non-deducibility of strategies} \cite{DBLP:conf/sp/WittboldJ90} which informally requires, that every possible observation made by the lower player is compatible with every possible \emph{strategy} for the high player.

Recall the NDS formulation in \HyperATLS:
$$
\neg\left(\exists \pi_1. \llangle \agent_H \rrangle \pi_2. \ltlg (\bigwedge_{a \in L} a_{\pi_1} \leftrightarrow a_{\pi_2}) \to \lozenge (\bigvee_{a \in O} a_{\pi_1} \not\leftrightarrow a_{\pi_2})\right)
$$
Assume the negation of this formula holds.
In this case there exists a trace $\pi_1$ such that $\agent_H$ has a stratgy to avoid the output on $\pi_1$ (no matter how hard the no-determinism tries). 
Given the same low-security inputs (as on $\pi_1$) the high-security player can thus enforce to different outputs, thereby communicating information to an untrusted source. 
The example system described before does not satisfy this formula (and is thus leaky).
(Note that in the example above no low-security inputs are present). 
We can show that non-deducibility of strategies does implies GNI.
The idea is, that whenever there is a trace of high-security inputs that avoids a specific output, then there always is a strategy that always chooses the inputs according to the trace and ignores all further information.
Viewing inputs as a trace is thus a special case of viewing inputs as given by strategies.

\newtheorem*{reNDS}{Restatement of Lemma \ref{lem:NDS}}

\begin{reNDS}
	If a system $\mathcal{G}$ does satisfy $\mathit{NDS}$ then it satisfies $\mathit{GNI}$.
\end{reNDS}
\begin{proof}
	We show that every system that does not satisfy $\mathit{GNI}$ also does not satisfies $\mathit{NDS}$, i.e., the contraposition.
	So lets assume $S \not\models \mathit{GNI}$, i.e., there exists traces $\pi_H$, $\pi_{O}$ such that \emph{no} trace agrees with the high-inputs of $\pi_H$ and outputs of $\pi_{O}$.
	In $\mathit{NDS}$ we take $\pi_1 := \pi_{O}$. The strategy for $\agent_H$ we selects the high-inputs according to the trace $\pi_H$.
	As there is no trace that combines the inputs from $\pi_H$ and outputs from $\pi_{O}$, every trace that conforms with the inputs chosen by $\agent_H$ and agrees with $\pi_1$ on low-security inputs, must differ in at least position from the low outputs in $\pi_1$, so $\mathit{NDS}$ is satisfied. 
\end{proof}

\section {Additional Material for \refSection{mc} - HyperATL*~Model Checking}

In this section we provide the missing construction of the alternating automaton for the case where $\varphi = [\llangle A_1 \rrangle \pi_1. \cdots \llangle A_k \rrangle \pi_k. \psi$. 
Again assume that we inductively created an alternating automaton $\mathcal{A}_\psi$ that is $\mathcal{G}$-equivalent to $\psi$.
$\mathcal{A}_\psi$ is an automaton over $S^{n+k}$. 
We wish to reduce the alphabet to $S^{n}$.
Assume $\mathcal{G} = (S, s_0, \agents, \mathscr{M}, \delta, \{d_s\}_{s \in S},  \atomic, L)$ and the deterministic parity automaton $\mathcal{A}^\mathit{det}_\psi = (Q, q_0, \lambda : Q \times S^{n+k} \to Q, c)$ constructed from $\mathcal{A}_\psi$ via \refTheo{alt1}.
For any set of agents $A$ we define $\mathit{level}(A, i) := A \cap d^{-1}(i)$. Conversely $\mathit{level}(\overline{A}, i) := (\agents \setminus A) \cap d^{-1}(i)$.
We now define $\mathcal{A}_\varphi := (S^k \times Q \cup \{q_\mathit{init}\}, q_\mathit{init}, S^n, \rho, c')$ where $\rho$ is defined by
{\small\begin{align*}
		\rho\left((q_\mathit{init}, q), [s_1, \cdots, s_n]\right) = &\bigvee\limits_{\sigma^1_1 : \mathit{level}(A_1, 1) \to \mathscr{M}} \cdots \bigvee\limits_{\sigma^k_1 : \mathit{level}(A_k, 1) \to \mathscr{M}} \\
		&\bigwedge\limits_{\sigma'^1_1 : \mathit{level}(\overline{A_1},  1) \to \mathscr{M}} \cdots \bigwedge\limits_{\sigma'^k_1 : \mathit{level}(\overline{A_k}, 1) \to \mathscr{M}} \\
		&\cdots \\
		&\bigvee\limits_{\sigma^1_m : \mathit{level}(A_1, m) \to \mathscr{M}} \cdots \bigvee\limits_{\sigma^k_m : \mathit{level}(A_k, m) \to \mathscr{M}} \\
		&\bigwedge\limits_{\sigma'^1_m : \mathit{level}(\overline{A_1},  m) \to \mathscr{M}} \cdots \bigwedge\limits_{\sigma'^k_m: \mathit{level}(\overline{A_k}, m) \to \mathscr{M}} \\
		&\big(\delta_{s_n^\circ}( \sum_{i=1}^{m} \sigma^1_i  + \sigma'^1_1), \cdots, \delta_{s_n^\circ}( \sum_{i=1}^{m} \sigma^k_i  + \sigma'^k_1), \lambda(q_0, [s_1, \cdots, s_n, s_n^\circ, \cdots, s_n^\circ])   \big)\\
\end{align*}}%
{\small\begin{align*}
		\rho\left(([t_1, \cdots, t_k], q), [s_1, \cdots, s_n]\right) = &\bigvee\limits_{\sigma^1_1 : \mathit{level}(A_1, 1) \to \mathscr{M}} \cdots \bigvee\limits_{\sigma^k_1 : \mathit{level}(A_k, 1) \to \mathscr{M}} \\
		&\bigwedge\limits_{\sigma'^1_1 : \mathit{level}(\overline{A_1},  1) \to \mathscr{M}} \cdots \bigwedge\limits_{\sigma'^k_1 : \mathit{level}(\overline{A_k}, 1) \to \mathscr{M}} \\
		&\cdots \\
		&\bigvee\limits_{\sigma^1_m : \mathit{level}(A_1, m) \to \mathscr{M}} \cdots \bigvee\limits_{\sigma^k_m : \mathit{level}(A_k, m) \to \mathscr{M}} \\
		&\bigwedge\limits_{\sigma'^1_m : \mathit{level}(\overline{A_1},  m) \to \mathscr{M}} \cdots \bigwedge\limits_{\sigma'^k_m: \mathit{level}(\overline{A_k}, m) \to \mathscr{M}} \\
		&\big(\delta_{t_1}( \sum_{i=1}^{m} \sigma^1_i  + \sigma'^1_1), \cdots, \delta_{t_k}( \sum_{i=1}^{m} \sigma^k_i  + \sigma'^k_1), \lambda(q, [s_1, \cdots, s_n, t_1, \cdots, t_k])   \big)
\end{align*}}
where we define $s_n^\circ := s_n$ of $n \geq 1$ and $s_n^\circ = s_0$ (the initial state of $\mathcal{G}$) otherwise.
We define $c'([t_1, \cdots, t_k], q) = c(q)$ and let $c'(q_\mathit{init})$ be chosen arbitrarily (it does not matter as we visit $q_\mathit{init}$ only once).

\subsection{Correctness Proof}

In this section we give a proof of \refProp{MCequiv}.
For sake of simplicity we focus on the non-parallel-composition semantics.
The proof extends easily to the automaton defined for formulas of the form $\varphi = [\llangle A_1 \rrangle \pi_1. \cdots \llangle A_k \rrangle \pi_k.] \psi$. 
We also assume that the game structure is CGS (i.e., only has one stage) as this simplifies notation.
The proof can, however, easily be extended. 
We prove \refProp{MCequiv} by structural induction on $\varphi$.
The case of \LTL~operators and boolean connectives is similar to the correctness of the standard \LTL~to APA translation \cite{DBLP:conf/lics/MullerSS88,DBLP:books/sp/cstoday95/Vardi95}.
We can therefore focus on the case where $\varphi = \llangle A \rrangle \pi. \psi$.
Let $\pi_1, \cdots, \pi_n$ be the path variables that are bound outside of $\varphi$.
We show that for all paths $t_1, \cdots, t_n \in S^\omega$, we have $[\pi_i \mapsto t_i]_{i = 1}^n \models_{\mathcal{G}} \varphi$ if and only if $\mathit{zip}(t_1, \cdots, t_n) \in \mathcal{L}(\mathcal{A}_\varphi)$.
For our proof we work with a modified version of the constructed $\mathcal{A}_\varphi$, that is more verbose and allows for an easier extraction of strategies.

\subparagraph{Alternative Construction}

Let $\mathcal{G} = (S, s_0, \agents, \mathscr{M}, \delta, \atomic, L)$ be a CGS and the already determiniszed automaton $\mathcal{A}^\mathit{det}_\psi = (Q, q_0, \lambda : Q \times S^{n+1} \to Q, c)$ obtained from the inductively constructed $\mathcal{A}_\psi$.
Instead of choosing the statespace as $S \times Q \cup \{q_\mathit{init}\}$ we consider an automaton with an extended statespace that is equivalent to $\mathcal{A}_\varphi$ but makes the disjunctive choices explicit.
We define the alternative automaton 
$$\tilde{\mathcal{A}}_\varphi = (S \times Q \times (A \to \mathscr{M}) \times (\overline{A} \to \mathscr{M}) \cup \{q_\mathit{init}\}, q_\mathit{init}, S^n, \rho, c')$$ 
where $\rho$ is defined by
{\small\begin{align*}
		\tilde{\rho}(q_\mathit{init}, [s_1, \cdots, s_n]) &= \bigvee\limits_{\sigma : A \to \mathscr{M}} \bigwedge\limits_{\sigma' : \overline{A} \to \mathscr{M}}   
		\big(\delta_{s^\circ_n}( \sigma  + \sigma'), \lambda(q_0, [s_1, \cdots, s_n, s^\circ_n]), \sigma, \sigma'\big)\\
		\tilde{\rho}\left((s, q, \_, \_), [s_1, \cdots, s_n]\right) &= \bigvee\limits_{\sigma : A \to \mathscr{M}} \bigwedge\limits_{\sigma' : \overline{A} \to \mathscr{M}}  
		\big(\delta_s( \sigma  + \sigma'), \lambda(q, [s_1, \cdots, s_n, s]), \sigma, \sigma' \big)
\end{align*}}%
Here $\_$ indicates that we do not care about the argument.
As before, we define $c'(s, q, \sigma_1, \sigma_1', \cdots, \sigma_m, \sigma_m') = c(q)$.
Note that focus on a CGS opposed to a MSCGS  simplifies notation as we only need to keep track of one disjunction.

The modified automaton $\tilde{\mathcal{A}}_\varphi$ behaves similarly to $\mathcal{A}_\varphi$ and in particular the behaviour on the $S \times Q$-components of the state-space is identical.
The only difference is that each state (apart from the initial state) now records the move vectors used to reach this state.
It is easy to see:

\begin{lemma}\label{lem:altConst}
	The languages of $\mathcal{A}_\varphi$ and $\tilde{\mathcal{A}}_\varphi$ are identical.
\end{lemma}

We can now show both direction separately.
\begin{lemma}\label{lem:dir1}
	If $\mathit{zip}(t_1, \cdots, t_n) \in \mathcal{L}(\mathcal{A}_\varphi)$ then $[\pi_i \mapsto t_i]_{i = 1}^n \models_{\mathcal{G}} \varphi$
\end{lemma}
\begin{proof}
	As $\mathit{zip}(t_1, \cdots, t_n) \in \mathcal{L}(\mathcal{A}_\varphi)$ we have $\mathit{zip}(t_1, \cdots, t_n) \in \mathcal{L}(\tilde{\mathcal{A}}_\varphi)$ by \refLemma{altConst}.
	Let $(T, r)$ be an accepting run of $\tilde{\mathcal{A}}_\varphi$ on $\mathit{zip}(t_1, \cdots, t_n)$.
	The proof-idea is now to use the disjunctive part chosen in $(T, r)$ as the move selection of a strategy. The strategy will thus simulate any given prefix in $(T, r)$ and select the move that was selected disjunctively in the transition function of $\tilde{\mathcal{A}}_\varphi$.
	
	We construct a strategy $f_\agent$ for each $\agent \in A$ as follows: Given any finite play $u = s_0, \cdots, s_k$, we check if there exists a node $\tau$ in $(T, r)$ such that the path to $\tau$ is labelled by $u$, i.e., $r(\epsilon), r(\tau[0, 0]), r(\tau[0, 1]), \cdots, r(\tau[0, |\tau|-1]) = q_{\mathit{init}},(s_1, \_, \_, \_), \cdots (s_k, \_, \_, \_)$. Note that in this case $|\tau| = |u|$. Here $\_$ means that we ignore the value in that position. 
	If no such path exists, select a arbitrary move (we later see that the constructed strategy will ever reach such a situation). 
	Otherwise let $r(\tau) = (s_k, q, \_, \_)$ where $q$ is an state of the automaton.
	By construction of $\tilde{\mathcal{A}}_\varphi$ we have that the children of $\tau$ satisfy. 
	$$\bigvee\limits_{\sigma : A \to \mathscr{M}} \bigwedge\limits_{\sigma' : \overline{A} \to \mathscr{M}}   
		\big(\delta_{s_k}( \sigma  + \sigma'), \lambda(q, [t_1(|\tau|), \cdots, t_n(|\tau|), s_k]), \sigma, \sigma'\big)$$
	There must exist a child of $\tau$ and a move vector $\sigma : A \to \mathscr{M}$ such that for all $\sigma' : \overline{A} \to \mathscr{M}$ there is a node labelled by $(\delta_{s_k}( \sigma  + \sigma'), \lambda(q, [t_1(|\tau|), \cdots, t_n(|\tau|), s_k], \sigma, \sigma')$.
	For any $\agent \in A$ now define $f_\agent(u) := \sigma(\agent)$.
	I.e., choose the move selected in the strategy profile for $A$.
	Note that as \emph{all} $\sigma' : \overline{A} \to \mathscr{M}$ also appeard as a child, all possible moves of the adversarial agents also result in a transition in the tree.
	
	It is easy to see that for the resulting strategy $F_A$ we get that for all $u \in \mathit{out}(\mathcal{G}, t_n^\circ, F_A)$ there exist a path in $(T, r)$ labelled with $q_\mathit{init} (u(1), q_1, \_, \_) (u(2), q_2, \_, \_) \cdots$.
	By definition of $\tilde{\rho}$, the sequence of automaton state $q_1, q_2, \cdots$ is the unique run of $\mathcal{A}^\mathit{det}_\psi$ on $\mathit{zip}(t_1, \cdots, t_n, u)$.
	As $(T, r)$ is accepting this sequence of automata states is accepting, we thus get that $\mathit{zip}(t_1, \cdots, t_n, u) \in \mathcal{A}_\psi$ and by assumption ($\mathcal{A}_\psi$ is $\mathcal{G}$-equivalent to $\psi$), we get that $[\pi_i \mapsto t_i]_{i = 1}^n \cup [\pi \mapsto u] \models_{\mathcal{G}} \psi$.
	As this holds for all $u \in \mathit{out}(\mathcal{G}, t_n^\circ, F_A)$, $F_A$ is a winning set of strategies and $[\pi_i \mapsto t_i]_{i = 1}^n \models_{\mathcal{G}} \varphi$ by the semantics of \HyperATLS.
\end{proof}

\begin{lemma}\label{lem:dir2}
	If $[\pi_i \mapsto t_i]_{i = 1}^n \models_{\mathcal{G}} \varphi$ then $\mathit{zip}(t_1, \cdots, t_n) \in \mathcal{L}(\mathcal{A}_\varphi)$
\end{lemma}
\begin{proof}
	Let $F_A$ be a winning strategy for the agents in $A$, i.e., for all $u \in \mathit{out}(\mathcal{G}, t_n^\circ, F_A)$,  $[\pi_i \mapsto t_i]_{i = 1}^n \cup [\pi \mapsto u]\models_{\mathcal{G}} \psi$.
	We construct an accepting run $(T, r)$ for $\mathit{zip}(t_1, \cdots, t_n)$. We construct this tree incrementally (fomally by induction on the depth).
	
	For the root the label is obvious.
	Now let $\tau \in T$ with $r(\tau) = (s, q)$ be any node in this so far constructed automaton. Let $r(\epsilon)r(\tau[0,0])r(\tau[0, 1]) \cdots r(\tau[0, |\tau|-1])$ be the path of node leading to $\tau$ and $q_{\mathit{init}}(s_1, q_1) \cdots (s_k, q_k) (s, q)$ the label of those nodes.
	Now design a move vector $\sigma : A \to \mathscr{M}$ as follows: $\sigma(\agent) := f_\agent(s_n^\circ,s_1, \cdots, s_k, s)$ for each $\agent \in A$.
	We then construct the children of $\tau$ by iterating over all of the $\sigma' : \overline{A} \to \mathscr{M}$ move vectors and construct a node labelled with $(\delta(s, \sigma + \sigma'), \lambda(q, [t_1(|\tau|), \cdots, t_n(|\tau|), s])$.
	By construction of the transition function of $\mathcal{A}_\varphi$, those children satisfy the transition-relation.
	
	The constructed tree $(T, r)$ is thus a run on $\mathit{zip}(t_1, \cdots, t_n)$.
	We now claim that $(T, r)$ is accepting. Consider any infinite path in this tree labelled $q_\mathit{init} (s_1, q_1) (s_2, q_2), \cdots$. By construction of the tree it is easy to see that the path $u = t_n^\circ, s_1, s_2, \cdots$ is contained $\mathit{out}(\mathcal{G}, t_n^\circ, F_A)$. 
	As $F_A$ is winning, we thus have that $[\pi_i \mapsto t_i]_{i = 1}^n \cup [\pi \mapsto u] \models_{\mathcal{G}} \psi$ and by our hypothesis (we assume that $\mathcal{A}_\psi$ is $\mathcal{G}$-equivalent to $\psi$) we get $\mathit{zip}(t_1, \cdots, t_n, u) \in \mathcal{L}(\mathcal{A}_\psi)$.
	Now $\mathcal{A}_\varphi$ simulates the automaton $\mathcal{A}_\psi$ on the input sequence augmented with the states from $u$.
	In particular, the automaton on the infinite path in $(T, r)$ is the same sequence as in accepting run on $\mathit{zip}(t_1, \cdots, t_n, u)$.
	We therefore have that $(T, r)$ is \emph{accepting} and thus that $\mathit{zip}(t_1, \cdots, t_n) \in \mathcal{L}(\mathcal{A}_\varphi)$.
\end{proof}

We can finally prove \refProp{MCequiv}.

\newtheorem*{reMCequiv}{Restatement of Proposition \ref{prop:MCequiv}}

\begin{reMCequiv}
	$\mathcal{A}_\varphi$ is $\mathcal{G}$-equivalent to $\varphi$.
\end{reMCequiv}
\begin{proof}
	Follows by induction and directly from the correctness of the \LTL~to APA construction \cite{DBLP:conf/lics/MullerSS88,DBLP:books/sp/cstoday95/Vardi95}, \refLemma{dir1} and \refLemma{dir2}.
\end{proof}

\newtheorem*{reupper}{Restatement of Theorem \ref{theo:upperBound}}
\begin{reupper}
	Model checking of a \HyperATLS~formula with $k$ complex and $l$ simple quantifiers is in $(2k + l)$-\EXPTIME. If $l > 0$ and the formula is linear it is also in $(2k+l-1)$-\EXPSPACE~(both in size of the formula). 
\end{reupper}
\begin{proof}
	Let $\varphi$ be any \HyperATLS~in negation-normal form formula with $k$-complex and $l$ simple quantifiers.
	By trivial induction on $\varphi$ we get that the constructed alternating automaton is of size $\mathcal{T}_2(2k + l, \mathcal{O}(|\varphi|))$, as each complex quantifier requires a determinisation whereas simple quantifier only requires a non-determinisation (\refTheo{alt1}).
	
	As a upper bound we thus get is that the MC is problem is in $(2k + l)$-\EXPTIME, as an alternating parity automaton can be checked for emptiness in polynomial time (assuming a fixed set of colours). 
	
	If we have $l > 0$ and the formula is linear (the analogous to the syntactic restriction of \HyperCTLS vs \HyperLTL) we can actually do better.
	There are two cases:
	Either the simple quantifier (of which at least one exists) occurs in the scope of a complex quantifier in which case we do not need to determinise an alternating automaton but merely a non-deterministic or universal one (e.g. for a formula of the form $\llangle A \rrangle \pi_1. \exists \pi_2. \psi$ the resulting automaton $\mathcal{A}_{\exists \pi_2. \psi}$ is already non-deterministic, so we can determinise it with a single exponential blow-up. ).
	In this case we get a even better upper bound of $(2k + l - 1)$-\EXPTIME.
	In the other case the simple quantifier occurs as the outermost quantifier in which case we get $(2k + l-1)$-\EXPSPACE~upper bound as emptiness of non-deterministic automata is in \NLOGSPACE~\cite{DBLP:journals/iandc/VardiW94} (and we can construct the automaton very space-efficiently, i.e., can check emptiness without writing down the entire final automaton).

	Note that the requirement that the formula is linear is to strict. We could derive a more precise bound by defining a notion similar to the alternation depth in a formula (see \cite{DBLP:conf/cav/FinkbeinerRS15} for details).
	Every alternation between quantifiers counts $1$ and every complex quantifier $2$.
	Similar to \cite{DBLP:conf/cav/FinkbeinerRS15}, existential quantification in the left of an until or right of a release or universal quantification in the right of an until or left of a release contributes also to the counter. 
\end{proof}

\section{Additional Proofs for \refSection{lb} - Lower Bounds For HyperATL* Model Checking}

In this section we give a proof of \refTheo{lb}.
We prove the lower bound by casting the termination of space-bounded Turing machines into the \HyperATLS~model checking.

The basic idea is to encode a yardstick, i.e., a formula that specified that two atomic propositions occur a fixed number of steps apart.
Given a yardstick of length $f(n)$ we can encode $f(n)$-space bounded Turing machines as shown in \cite{DBLP:phd/dnb/Rabe16}.
The crucial insight to get the lower bound is thus to create a yardstick that increases by two exponents with each strategic quantifier. 
In the end we will show that checking a formula of the form  
$$\exists \pi. \llangle A \rrangle \pi_1. \cdots \llangle A \rrangle \pi_k. \varphi$$
is already $2k$-\EXPSPACE-hard. 
For $k, n \in \mathbb{N}$ we define the following tower $2^{[k]}(n)$ as follows:
\begin{align*}
	2^{[0]}(n) &= n\\
	2^{[k+1]}(n) &= 2^{2^{2^{[k]}(n)}} \cdot 2^{[k]}(n) \cdot 2^{2^{[k]}(n)}
\end{align*}
It is easy to see that for every $k \geq 1$ we have $2^{[k]}(n) \geq \mathcal{T}_2(2k, n)$ for every $n$.
We will show how to encode a yardstick of length $2^{[k]}(n)$ as a formula of the form $\llangle A \rrangle \pi_1. \cdots \llangle A \rrangle \pi_k. \varphi$. By appending a single existential quantifier we can then encode Turing machines of that length.
We recommand the reader to have a look at the lower bound proof in \cite{DBLP:phd/dnb/Rabe16} for a more gentle introduction into the proof idea.

\subparagraph{Structure}
This section is structured as follows.
We begin by giving a proof intuition for the base case, i.e., encoding a yardstick of length $2^{[1]}(n) = 2^{2^n} \cdot n \cdot 2^n$ with a single strategic quantifier. 
We then give a intuition on how to stack these counter.
Finally we give a formal description of the resulting formula.

\subsection{Intuition for the Base Case}

In  \LTL~it is straightforward to encode a counter of length $n$ and thus express a sequence of length $2^n$ (more precisely length $n \cdot 2^n$).
We now show how the strategic behaviour in \HyperATLS~allows us to encode a doubly exponential counter, i.e., a counter of already exponential length.
The game structure in our reduction describes a game between two players, one of which will be controlled via a strategy (called the $\exists$-player) and one will be unrestricted (called the $\forall$-player). 
Our formula will require that the $\exists$-player produces a counter of exponential length and use the $\forall$-player to control that this counter is correct.
The $\exists$-player should produce two counters simultaneously. 
The first counter is specified via atomic proposition $a$ and of length $n$ (i.e., count from $0$ to $2^n - 1$). We call this the $a$-counter.
The second counter is via proposition $b$ and of length $2^n$ (i.e., count from $0$ to $2^{(2^n)}-1$), called the $b$-counter. 
The structure of the counters produced by the system should be the following:

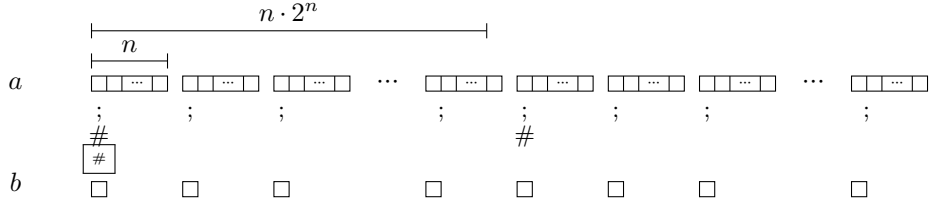
\begin{center}
	\begin{tikzpicture}
		
		\node[] at (-1, 0.1) () {$a$};

		\draw[draw=black] (0, 0) rectangle (1, 0.2);
		\draw[draw=black] (0.2, 0) -- (0.2, 0.2);
		\draw[draw=black] (0.4, 0) -- (0.4, 0.2);
		\node[] at (0.6, 0.1) () {\tiny$\cdots$};
		\draw[draw=black] (0.8, 0) -- (0.8, 0.2);
		
		\draw[draw=black] (1.2, 0) rectangle (2.2, 0.2);
		\draw[draw=black] (1.4, 0) -- (1.4, 0.2);
		\draw[draw=black] (1.6, 0) -- (1.6, 0.2);
		\node[] at (1.8, 0.1) () {\tiny$\cdots$};
		\draw[draw=black] (2, 0) -- (2, 0.2);
		
		\draw[draw=black] (2.4, 0) rectangle (3.4, 0.2);
		\draw[draw=black] (2.6, 0) -- (2.6, 0.2);
		\draw[draw=black] (2.8, 0) -- (2.8, 0.2);
		\node[] at (3, 0.1) () {\tiny$\cdots$};
		\draw[draw=black] (3.2, 0) -- (3.2, 0.2);
		
		\node[] at (3.9, 0.1) () {$\cdots$};
		
		\draw[draw=black] (4.4, 0) rectangle (5.4, 0.2);
		\draw[draw=black] (4.6, 0) -- (4.6, 0.2);
		\draw[draw=black] (4.8, 0) -- (4.8, 0.2);
		\node[] at (5, 0.1) () {\tiny$\cdots$};
		\draw[draw=black] (5.2, 0) -- (5.2, 0.2);

		\node[] at (0.1, -0.3) () {\small$;$};
		\node[] at (1.3, -0.3) () {\small$;$};
		\node[] at (2.5, -0.3) () {\small$;$};
		\node[] at (4.6, -0.3) () {\small$;$};
		
		\node[] at (0.1, -0.6) () {\small$\#$};
		
		
		\draw[draw=black] (0+5.6, 0) rectangle (1+5.6, 0.2);
		\draw[draw=black] (0.2+5.6, 0) -- (0.2+5.6, 0.2);
		\draw[draw=black] (0.4+5.6, 0) -- (0.4+5.6, 0.2);
		\node[] at (0.6+5.6, 0.1) () {\tiny$\cdots$};
		\draw[draw=black] (0.8+5.6, 0) -- (0.8+5.6, 0.2);
		
		\draw[draw=black] (1.2+5.6, 0) rectangle (2.2+5.6, 0.2);
		\draw[draw=black] (1.4+5.6, 0) -- (1.4+5.6, 0.2);
		\draw[draw=black] (1.6+5.6, 0) -- (1.6+5.6, 0.2);
		\node[] at (1.8+5.6, 0.1) () {\tiny$\cdots$};
		\draw[draw=black] (2+5.6, 0) -- (2+5.6, 0.2);
		
		\draw[draw=black] (2.4+5.6, 0) rectangle (3.4+5.6, 0.2);
		\draw[draw=black] (2.6+5.6, 0) -- (2.6+5.6, 0.2);
		\draw[draw=black] (2.8+5.6, 0) -- (2.8+5.6, 0.2);
		\node[] at (3+5.6, 0.1) () {\tiny$\cdots$};
		\draw[draw=black] (3.2+5.6, 0) -- (3.2+5.6, 0.2);
		
		\node[] at (3.9+5.6, 0.1) () {$\cdots$};
		
		\draw[draw=black] (4.4+5.6, 0) rectangle (5.4+5.6, 0.2);
		\draw[draw=black] (4.6+5.6, 0) -- (4.6+5.6, 0.2);
		\draw[draw=black] (4.8+5.6, 0) -- (4.8+5.6, 0.2);
		\node[] at (5+5.6, 0.1) () {\tiny$\cdots$};
		\draw[draw=black] (5.2+5.6, 0) -- (5.2+5.6, 0.2);

		\node[] at (0.1+5.6, -0.3) () {\small$;$};
		\node[] at (1.3+5.6, -0.3) () {\small$;$};
		\node[] at (2.5+5.6, -0.3) () {\small$;$};
		\node[] at (4.6+5.6, -0.3) () {\small$;$};
		
		\node[] at (0.1+5.6, -0.6) () {\small$\#$};

		\node[] at (0.1, -0.9) () {\tiny$\boxed{\#}$};

		
		\node[align=center] at (-1, -1.2) () {$b$};
		
		\draw[draw=black] (0, -1.2) rectangle (0.2, -1.4);
		\draw[draw=black] (1.2, -1.2) rectangle (1.4, -1.4);
		\draw[draw=black] (2.4, -1.2) rectangle (2.6, -1.4);
		\draw[draw=black] (4.4, -1.2) rectangle (4.6, -1.4);
		
		\draw[draw=black] (0+5.6, -1.2) rectangle (0.2+5.6,-1.4);
		\draw[draw=black] (1.2+5.6, -1.2) rectangle (1.4+5.6, -1.4);
		\draw[draw=black] (2.4+5.6, -1.2) rectangle (2.6+5.6, -1.4);
		\draw[draw=black] (4.4+5.6, -1.2) rectangle (4.6+5.6, -1.4);


		\draw[draw=black] (0, 0.3) -- (0, 0.5);
		\draw[draw=black] (0, 0.4) --node[above] {$n$} (1, 0.4);
		\draw[draw=black] (1, 0.5) -- (1, 0.3);

		\draw[draw=black] (0, 0.7) -- (0, 0.9);
		\draw[draw=black] (0, 0.8) --node[above] {$n \cdot 2^n$} (5.2, 0.8);
		\draw[draw=black] (5.2, 0.9) -- (5.2, 0.7);

	\end{tikzpicture}
\end{center}

The counter consist of atomic proposition $a, b,;, \#$ and $\boxed{\#}$. The $a$-counter is of length $n$ and separated via $;$ (that is $;$ occurs every $n$-positions). 
Each box in the row labelled $a$ represents counter bit where $a$ being set is understood as $1$ and $a$ not being set as $0$.
We call a word between two $;$ an $a$-count. The $a$-counter should continuously count from $0$ to $2^n-1$ and then repeat. 
We mark the beginning of a new counter sequence (i.e., $\#$ occurs whenever the next $a$-count is $0$) by $\#$ (Note that $\#$ thus always aligns with $;$).

Proposition $b$ is used to encode the second counter of length $2^n$. Its bits are always placed at the beginning of the smaller counter (indicated via $;$). 
The occurrence of $b$ at positions that do not align with $;$are irrelevant for the counter. Between each occur of $\#$ there thus exists exactly $2^n$ positions where the $b$-proposition counts towards the counter. 
We call such a sequence a $b$-count. 
As for the $a$ counter, the $b$-counter should continuously count from $0$ to $2^{2^n}-1$ and then restart. 
Proposition $\boxed{\#}$ is used to indicate the start of a new $b$-counter, i.e., it holds whenever the next $b$-count value is $0$. 

If this counter is correct, then any two occurrences of $\#$ are exactly $n \cdot 2^n$ steps apart and each occurrences of $\boxed{\#}$ is exactly $2^{2^n} \cdot n \cdot 2^n$ (as each $b$ count has length $n \cdot 2^n$).
We can thus use the occurrence of $\boxed{\#}$ as a yardstick of doubly exponential length. 

Our CGS now (at least for now) allows the $\exists$-player to generated arbitrary words over the atomic proportions $a, b, ;, \#, \boxed{\#}$.
Note that the size of the CGS is thus constant. 

The first objective is to verify the correctness of the $a$-counter. As the length of this counter is linear, we can directly express its correctness in \LTL~using the next operator.
The tricky part is the verification of the $b$-counter, as checking the correctness of this counter requires the compromising at two positions that are roughly $2^n$ positions apart. 
The idea is now to use the $\forall$-player to challenge the correctness of his counter.
We modify the CGS to allow the $\forall$-player to set a special atomic propositions $\mathit{error}$ at each location where proposition ``$;$`` holds, i.e, at every possible bit of the $b$-counter.
The $\forall$-player should set this proposition to indicate that the $\exists$-player made a mistake in the $b$-counter at the location the flag has been set.

We use the fact that for two $m$-bit counter $\alpha_0, \cdots, \alpha_m$ and $\beta_0, \cdots, \beta_m$ (where $\alpha_m, \beta_m$ are the least significant bits) it holds that $\beta = (\alpha + 1) \, \mathit{mod} \, 2^m$ if and only if for every $i$, $\alpha_i = \beta_i$ iff $\alpha_j = 0$ for some $j > i$. Call this fact \textbf{(count)}.
The first error in the $b$ counter is thus the first position in $\beta$ that does not satisfy this formula. 

In our formula we require that the $\forall$-player sets $\mathit{error}$ at at most one position on a trace (otherwise the $\exists$-player wins). 
Having marked the position of the error we can now compare positions via a \LTL~formula.
We wish to compare the marked position with the same position in the previous $b$-count, i.e., if $\mathit{error}$ marks the $j$ position in a $b$-count we wish to compare with the $j$-th position in the \emph{previous} $b$-count. 
The crux is that we can use the $a$-counter to find this position. 
We define the formula
\begin{align}
	\mathit{prevPos} := &; \land \tag*{(1)} \\
	&(\neg \#) \mathcal{U} (\# \land ((\neg \#) \mathcal{U} \mathit{error})) \land \tag*{(2)}  \\ 
	&\bigwedge_{i \in \{0, \cdots, n-1\}} \ltlnext^i a \leftrightarrow \ltlg (\mathit{error} \Rightarrow \ltlnext^i a) \tag*{(3)}
\end{align}
Here (1) states that the position aligns with a $;$-marker, i.e., a bit of the $b$-counter.
(2) states the we are in the previous $b$-counter, i.e., there is exactly one occurrence of $\#$ until we reach the $\mathit{error}$-position. 
(3) requires that for the next $n$ positions, the value of $a$ agrees exactly with the same letter in the $a$ counter after the $\mathit{error}$ location. 
$\mathit{prevPos}$ now holds exactly once on the trace (provided $\mathit{error}$ is set once and the $a$ counter is correct).
It holds at exactly the position that corresponds to the same position as $\mathit{error}$ in the previous $b$-word. 
As we can now compare both positions we can express that the counter at this positions is correct using \textbf{(count)} by the formula:
$$
(\ltlf \mathit{error}) \Rightarrow \left(\ltlg \left(\mathit{prevPos} \Rightarrow \left( (b \Leftrightarrow \ltlg (\mathit{error} \Rightarrow b))  \Leftrightarrow \mathit{before}(\#, ; \land \neg b) \right) \right)\right)
$$
This formula states that the two position pointed to via $\mathit{prevPos}$ and $\mathit{error}$ are correct, i.e., the $b$ value at both positions agree exactly if there is a zero bit on the $b$-counter after the $\mathit{prevPos}$ position (corresponding to \textbf{(count)}).
Here $\mathit{before} (\psi_1, \psi_2)$ expresses that $\psi_2$ should hold before $\psi_1$ holds, i.e., the zero-bit must occur before the end of the $b$-count in which $\mathit{prevPos}$ occurs.

The $\forall$-player can thus challenge every single bit of the $b$-counter so the only winning strategy for the $\exists$-player is to produce a correct counter where $a$ correctly counts from $0$ to $2^n-1$ and then repeats and $b$ counts from $0$ to $2^{2^n}-1$ and then repeats.
As $\boxed{\#}$ indicates the restart of the $b$ counter, each occurrence of $\boxed{\#}$ is exactly $2^{2^n} \cdot n \cdot 2^n$ positions apart.

\subsection{Intuition for the Inductive Case}

What remains now is to stack the counter to build a tower of exponentials.
To do so, we need to adjust the previous construction. 
Most importantly, the formula constructed in the previous section depend on $n$ as we use explicit nexts to compare consecutive positions (both in the verification of the $a$ and $b$ counter). 
In our inductive formula, we need to compare by using an inductively smaller yardstick. 
Our final formula has the from $\llangle \{\exists\} \rrangle \pi_k. \cdots \llangle \{\exists\} \rrangle \pi_1. \varphi$, where each trace $\pi_i$ should contain a counter/yardstick of length $2^{[i]}(n)$.
For the trace $\pi_1$ we already saw how to verify this above. 
For each $\pi_i$ where $i > 1$ we will use the yardstick on $\pi_{i-1}$ (which we can assume to be correct) to verify the counter on $\pi_i$.

To compare any possible positions on $\pi_i$ we modify the CGS such that the $\forall$-player can decide when to start the counter/yardstick. 
Any winning strategy for the $\exists$-player thus permits counter starting at every possible timepoint. As $\pi_{i-1}$ is resolved \emph{after} $\pi_i$ we can thus compare and check arbitrary positions. 

Let $\tilde{n} := 2^{[i-1]}(n)$ i.e., the yardstick length of the trace $\pi_{i-1}$.
We now sketch the high-level idea:
As before the $\exists$-player should produce a $\tilde{n}$-bit $a$ counter and a $2^{\tilde{n}}$-bit $b$ counter. 
The verification of $a$-counter is comparatively easy. As the $\forall$-player can start the counter on $\pi_{i-1}$ at any time, we can verify all positions of the counter to be correct.
In particular, the $\boxed{\#}$ proposition on $\pi_{i-1}$ occurs every $\tilde{n}$-steps and (as the $\forall$-player can determine the start) at every possible position.
This is similar to the universal quantification over yardsticks from \cite{DBLP:phd/dnb/Rabe16}.

To verify the $b$-counter, the $\forall$-player can set the $\mathit{error}$-proposition at a location he deems wrong.
Unlike in the base case, we can however not use the formula $\mathit{prevPos}$ to identify the same position in the previous $b$-count
(Note that the size of $\mathit{prevPos}$ depends on $n$.)
The trick is now to use the trace $\pi_{i-1}$ to identify this position. 
The modified CGS has thus two "modes". The $\forall$-player can either decide to start a yardstick/counter any time he wishes (this yardstick is used to verify the spacing and $a$-counter on trace $\pi_i$), or he can mark a positions if he thinks he has found an error on trace $\pi_i$.
In this case he sets the $\mathit{error}$ proposition on trace $\pi_i$ as before and sets a proposition $\mathit{errorStart}$ on trace $\pi_{i-1}$ at the same position in the previous $b$-count (i.e., the position where we used $\mathit{prevPos}$ in the base case). 
We can now check the correctness by comparing the positions pointed to by $\mathit{errorStart}$ on $\pi_{i-1}$ and $\mathit{error}$ on $\pi_i$.
The CGS now looks as follows where the black box represents the previous CGS in which the $\exists$-player produces traces and the $\forall$-player can set the $\mathit{error}$-proposition.

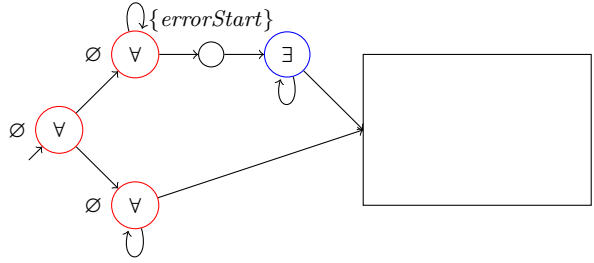
\begin{center}
	\begin{tikzpicture}
		\node[draw=red, circle, label=west:{\footnotesize$\emptyset$}] at (0,0) (ns) {\footnotesize$\forall$};
		
		\node[draw=red, circle, label=west:{\footnotesize$\emptyset$}] at (1,-1) (nostart) {\footnotesize$\forall$};
		
		\node[draw=red, circle, label=west:{\footnotesize$\emptyset$}] at (1,1) (startloop) {\footnotesize$\forall$};

		\node[draw=black, circle, label=north:{\footnotesize$\{\mathit{errorStart}\}$}] at (2,1) (start) {};
		
		\node[draw=blue, circle, label=north:{}] at (3,1) (startloopi) {\footnotesize$\exists$};

		\draw[->] (ns)+(-0.4, -0.4) -- (ns);
		\draw[->] (ns) -- (startloop);
		\draw[->] (ns) -- (nostart);
		
		\draw[->] (start) -- (startloopi);
		
		\path (startloopi) edge [loop below] node {} (startloopi);
		
		\path (startloop) edge [loop above] node {} (startloop);
		\path (nostart) edge [loop below] node {} (nostart);

		\draw[->] (startloop) -- (start);

		\draw[draw=black] (4, -1) rectangle (7,1);

		\draw[->] (nostart) -- (4, 0);
		\draw[->] (startloopi) -- (4, 0);

	\end{tikzpicture}
\end{center} 

If the $\forall$-player chooses the lower branch he can require the start of a counter at any possible timepoint.
If he found an error on trace $\pi_i$ he should choose the upper branch on trace $\pi_{i-1}$ and set the $\mathit{errorStart}$ position as discussed before. 

There is of course the possibility of the $\forall$-player to set the $\mathit{errorStart}$ flag on $\pi_{i-1}$ at the wrong position, i.e, at a position that does not correspond to the same position in the previous counter compared to $\mathit{error}$ on $\pi_i$. 
We therefore allow the $\exists$-player to question the correct placement of $\mathit{errorStart}$.
The $\exists$-player should demonstrate that the $a$-counter on $\pi_i$ starting at $\mathit{errorStart} 
_{\pi_{i-1}}$ and the $a$-counter on $\pi_i$ starting at $\mathit{error}_{\pi_i}$ do not agree.
He can do so by selecting a single bit that does not agree in both counters.
After the $\forall$-player has thus signaled $\mathit{errorStart}$ on $\pi_{i-1}$, the $\exists$-player can now determine when to start his yardstick of length $\tilde{n}$. He should start this yardstick such that the $a$-counter on $\pi_i$ does not agree at the specified position.
This yardstick (the starting position of which is fixed by the $\exists$-player) now signals $\boxed{\#}$ every $\tilde{n}$ steps. The first occurrence of $\boxed{\#}_{\pi_{i-1}}$ after $\mathit{errorStart}_{\pi_{i-1}}$ and the first occurrence of $\boxed{\#}_{\pi_{i-1}}$ after $\mathit{error}_{\pi_i}$ thus point to the same position in the two $a$-counts.
To prove the choice of the $\forall$-player wrong those bits should thus be different which implies that the $\mathit{errorStart}$-position was set incorrectly by the $\forall$-player. 

As $\pi_{i-1}$ is resolved after $\pi_i$ the forall player can verify all positions of the $a$-counter on $\pi_i$ (by choosing the lower branch in the model).
If he found an error on $\pi_i$ he can set the $\mathit{error}$ flag on $\pi_i$, then choose the upper branch in the model and set the $\mathit{errorStart}$ on $\pi_{i-1}$ at the correct position (i.e., the same position in the previous $b$ count). 
If he sets this position correctly, no matter how long the $\exists$-player prolongs the start of the counter, the position of the $a$-counter will always agree.
The only winning stratgy for the $\exists$-player is thus to produce a correct $a$ and $b$-counter.

\subsection{Formal Description}

We now give a full and formal description of the model and formula.

\subsubsection{The CGS}

The counter part of our model is a CGS where the $\exists$-player can generate arbitrary sequences over $\{a, b, ;, \#, \boxed{\#}\}$ and the $\forall$-player can choose the atomic proposition $\mathit{error}$ in each step.
The starting state is labelled by $\{;,\#, \boxed{\#}\}$.
The final CGS is the the following where the box is the CGS described above.
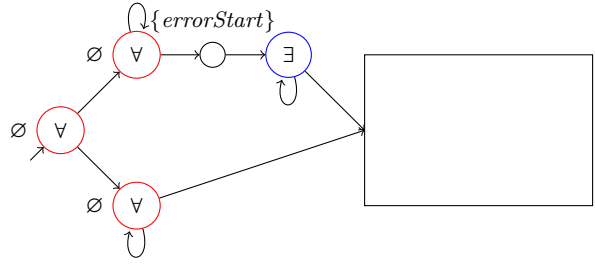
\begin{center}
	\begin{tikzpicture}
		\node[draw=red, circle, label=west:{\footnotesize$\emptyset$}] at (0,0) (ns) {\footnotesize$\forall$};

		\node[draw=red, circle, label=west:{\footnotesize$\emptyset$}] at (1,-1) (nostart) {\footnotesize$\forall$};

		\node[draw=red, circle, label=west:{\footnotesize$\emptyset$}] at (1,1) (startloop) {\footnotesize$\forall$};

		\node[draw=black, circle, label=north:{\footnotesize$\{\mathit{errorStart}\}$}] at (2,1) (start) {};
		
		\node[draw=blue, circle, label=north:{}] at (3,1) (startloopi) {\footnotesize$\exists$};

		\draw[->] (ns)+(-0.4, -0.4) -- (ns);
		\draw[->] (ns) -- (startloop);
		\draw[->] (ns) -- (nostart);
		
		\draw[->] (start) -- (startloopi);
		
		\path (startloopi) edge [loop below] node {} (startloopi);
		
		\path (startloop) edge [loop above] node {} (startloop);
		\path (nostart) edge [loop below] node {} (nostart);

		\draw[->] (startloop) -- (start);

		\draw[draw=black] (4, -1) rectangle (7,1);

		\draw[->] (nostart) -- (4, 0);
		\draw[->] (startloopi) -- (4, 0);
	\end{tikzpicture}
\end{center} 

The label of each node, indicates which player can choose the successor (i.e., the initial fragment is turn-based).
In the structure, the $\forall$-player can decide if he wants to start the counter at some point (by choosing the lower branch) of if he wants to include a $\mathit{errorStart}$ proposition before (by choosing the upper branch). 
In case he chooses the lower part, he can wait for any given number of time and then start the counter. 
Any possible strategy for $\exists$-player can therefore \emph{not} deciding when to start and must account for every possible starting point.
We can use these yardstick (which can start at any position) to very the $a$-counter and spacing of the next bigger counter.
If the $\forall$-player chooses the upper branch, the $\exists$-player can precisely determine when to start the yardstick. 
We will use this to allow the $\exists$-player to prove that the error-positions marked by the $\forall$-player are incorrect.

\subparagraph{The final Formula}

Our final formula has the from 
$$\llangle \{\exists\} \rrangle \pi_k. \cdots \llangle \{\exists\} \rrangle \pi_1. \bigwedge_{i \in \{1, \cdots, k\}} \mathit{correct}_i$$
where $\mathit{correct}_i$ encodes the that the counter on trace $\pi_i$ is correct. In particular the symbols $\boxed{\#}$ on trace $\pi_i$ are exactly $2^{[i]}(n)$ steps apart.
More concretely, every winning strategy for $\llangle \{\exists\} \rrangle \pi_k. \cdots \llangle \{\exists\} \rrangle \pi_1. \bigwedge_{i \in \{1, \cdots, k\}} \mathit{correct}_i$ only admins paths where the $\boxed{\#}_{\pi_k}$-occurs exactly $2^{[k]}(n)$-steps apart.
We distinguish between $\mathit{correct}_1$ and $\mathit{correct}_i$ for $i > 1$.

\subsubsection{The $i = 1$ case}

In the case of $i = 1$ we wish to guarantee that the counter on $\pi_i$ separates each $\boxed{\#}$ by exactly $2^{2^n}\cdot n \cdot 2^n$.
We follows the intuitions from before for constructing the base counter.
Formally,
\begin{align}
	\mathit{correct}_1 := 
	&\emptyset_{\pi_1} \mathcal{U} (\boxed{\#}_{\pi_1} \land \#_{\pi_1} \land ;_{\pi_1} ) \land\tag*{(1)}\\
	& \ltlg \left((\boxed{\#}_{\pi_1} \Rightarrow \#_{\pi_1}) \land (\#_{\pi_1} \Rightarrow ;_{\pi_1}) )\right) \land \ltlg ((;_{\pi_1}) \Rightarrow \ltlnext^n (;_{\pi_1}) \land \tag*{(2)}\\
	&\ltlg(\#_{\pi_1} \Leftrightarrow \ltlg^{\leq n - 1} \neg a_{\pi_1} ) \tag*{(3)}\\
	&\ltlg\left( (a_{\pi_1} \Leftrightarrow \ltlnext^n a_{\pi_1}) \Leftrightarrow \mathit{before}(;_{\pi_1}, \neg a_{\pi_1})  \right) \land \tag*{(4)}\\
	&\ltlg\left(\boxed{\#}_{\pi_1} \Leftrightarrow \left((;_{\pi_1} \Rightarrow \neg b_{\pi_1}) \mathcal{U} (\ltlnext \#_{\pi_1})\right) \right) \land \tag*{(5)}\\
	&\mathit{exactlyOnce}(\mathit{error}_{\pi_1}) \Rightarrow \mathit{refuteError}_1 \tag*{(6)}
\end{align}
(1) specifies that the first atomic proposition to appear must be $\boxed{\#}$, $\#$. Note that, as evident from the CGS, the counter must not start at position $0$ but may be started at any later point. 
We abbreviate $\emptyset$ to indicate that no proposition holds.
(2) specifies that the spacing of the separating proposition $;, \#$ and $\boxed{\#}$ is correct. In particular $;$ must be set every $n$ steps and whenever $\boxed{\#}$ is set so is $\#$ and when $\#$ is set so is $;$.
The exact placement of $\#$ and $\boxed{\#}$ is specified via the counter (see below).
(3) encodes that the a $\#$ letter appears exactly when the counter value of $a$ is zero. As (1) requires to start with a $\#$ this also requires the counter to be started at $0$.
(4) requires that the $a$-counter behaves as indented. 
Here we use the observation \textbf{(count)} from before.
We encode that $a$ agrees at two positions that are $n$ apart if and only if $\neg a$ holds for some lower bit, i.e., before $;$ indicates the start of the next counter. 
Here we use $\mathit{before}(\psi_1, \psi_2)$ which indicates that before $\psi_1$ holds for the first time, $\psi_2$ has been true.
Formally
$$\mathit{before}(\psi_1, \psi_2) = (\neg \psi_1) \mathcal{U} (\psi_2 \land \neg \psi_1)$$
(5) states that the $b$-counter is zero whenever $\boxed{\#}$ holds. Together with (1) this requires that the $b$-counter is initialized at $0$. 
Note that we only consider the positions of the $b$ counter to be those were $;$ holds. 
Finally (6) specifies the correct counting of the $b$-counter. As discussed before we use the $\forall$-player to mark any errors by setting $\mathit{error}$.
We only consider the case where the $\forall$-player set this flag exactly once.
$\mathit{exactlyOnce}(\psi)$ specifies that $\psi$ holds exactly once, i.e., 
$$\mathit{exactlyOnce}(\psi) := (\neg \psi) \mathcal{U} (\psi \land \ltlnext \ltlg(\neg \psi))$$
$\mathit{refuteError}_1$ then specifies that the $\mathit{error}$ flag set by the $\forall$-player does not mark an error, i.e., is a ``false alarm``.
To specify $\mathit{refuteError}_1$ we use a formula that specifies the same position in the previous $b$-count. 
Formally
\begin{align}
	\mathit{prevPos} := &;_{\pi_1} \land \tag*{(1)} \\
	&(\neg \#) \mathcal{U} (\#_{\pi_1} \land \ltlnext(\neg \#_{\pi_1}) \mathcal{U} \mathit{error}_{\pi_1}) \land \tag*{(2)}  \\ 
	&\bigwedge_{i \in \{0, \cdots, n-1\}} \ltlnext^i a_{\pi_1} \leftrightarrow \ltlg (\mathit{error} \Rightarrow \ltlnext^i a_{\pi_1}) \tag*{(3)}
\end{align}
Here (1) specifies that the position aligns with $;$, (2) that it occurs in the previous $b$-word, i.e., the position is separated from $\mathit{error}$ by exactly one $\#$. 
(3) states that the counter of $a$ (at the next $n$ positions) agrees with the one at the $\mathit{error}$-position. 
It is easy to see that if the $a$-counter is correct and $\mathit{error}$ occurs at some point, then $\mathit{prevPos}$ holds at exactly one location.
We can now express $\mathit{refuteError}_1$ via
$$
\ltlg \left(\mathit{prevPos} \Rightarrow \left( (b_{\pi_1} \Leftrightarrow \ltlg (\mathit{error}_{\pi_1} \Rightarrow b_{\pi_1}))  \Leftrightarrow \mathit{before}(\#_{\pi_1}, ;_{\pi_1} \land \neg b_{\pi_1}) \right) \right)
$$
This states that the $b$ counter is correct, i.e, it agrees at the two positions indicated via the formula $\mathit{prevPos}$ and the location of $\mathit{error}$ exactly if there is a $0$ after the $\mathit{prevPost}$ position. This can be express via $\mathit{before}(\#_{\pi_1}, ;_{\pi_1} \land \neg b_{\pi_1})$, i.e., at some position where $;$ holds (note that we only consider the $b$ proposition at those locations) we must have $\neg b$ before the current word end, i.e., before $\#$.

\subsubsection{The case for $i > 1$}

In the case where $i > 1$ we can no longer explicitly compare positions via the \LTL~next operator. 
In particular we can not verify the $a$-counter and the spacing of $;$ easily. 
Moreover, we can no longer use the formula $\mathit{prevPos}$ to give us the same position as $\mathit{error}$ in the previous word. 

Instead we use the indicatively defined yardstick on trace $\pi_{i-1}$. 
Note that, by construction of our CGS, the $\exists$-player can not decide when the counter should  be started. Every winning strategy on trace $\pi_{i-1}$ thus allows for correct counter starting at every possible position. 
As we can assume that the counter on $\pi_{i-1}$ is correct, we know that the $\boxed{\#}$-proposition on $\pi_{i-1}$ are exactly $\tilde{n} := 2^{[i-1]}(n)$ steps apart.
Our formula $\mathit{correct}_i$ now largely agrees with the case for $i = 1$. The only difference is that we use $\pi_{i-1}$ for measurements and use the proposition $\mathit{errorStart}$ on trace $\pi_{i-1}$ instead of the $\mathit{prevPos}$ formula constructed before. 
Define:
\begin{align}
	\mathit{correct}_i := 
	&\emptyset_{\pi_i} \mathcal{U} (\boxed{\#}_{\pi_i} \land \#_{\pi_i} \land ;_{\pi_i} ) \land\tag*{(1)}\\
	&\ltlg \left((\boxed{\#}_{\pi_i} \Rightarrow \#_{\pi_i}) \land (\#_{\pi_i} \Rightarrow ;_{\pi_i}) )\right) \land \ltlg (\boxed{\#}_{\pi_{i-1}} \land ;_{\pi_i} \Rightarrow \mathit{next}(\boxed{\#}_{\pi_{i-1}}, ;_{\pi_i}) ) \land \tag*{(2)}\\
	&\ltlg(\#_{\pi_1} \Leftrightarrow ((\neg a_{\pi_1}) \mathcal{U} (\ltlnext ;_{\pi_i})) ) \tag*{(3)}\\
	&\ltlg\left( \left(a_{\pi_i} \land \boxed{\#}_{\pi_{i-1}} \land \mathit{next}(\boxed{\#}_{\pi_{i-1}}, a_{\pi_i})\right) \Rightarrow  \mathit{before}( ;_{\pi_i}, \neg a_{\pi_i})  \right) \land \tag*{(4.1)}\\
	&\ltlg\left( \left(\neg a_{\pi_i} \land \boxed{\#}_{\pi_{i-1}} \land \mathit{next}(\boxed{\#}_{\pi_{i-1}}, \neg a_{\pi_i})\right) \Rightarrow  \mathit{before}( ;_{\pi_i}, \neg a_{\pi_i})  \right) \land \tag*{(4.2)}\\
	&\ltlg\left( \left(\neg a_{\pi_i} \land \boxed{\#}_{\pi_{i-1}} \land \mathit{next}(\boxed{\#}_{\pi_{i-1}}, a_{\pi_i})\right) \Rightarrow  a_{\pi_1} \mathcal{U} (\ltlnext ;_{\pi_i}) \right) \land \tag*{(4.3)}\\
	&\ltlg\left( \left(a_{\pi_i} \land \boxed{\#}_{\pi_{i-1}} \land \mathit{next}(\boxed{\#}_{\pi_{i-1}}, \neg a_{\pi_i})\right) \Rightarrow  a_{\pi_1} \mathcal{U} (\ltlnext ;_{\pi_i}) \right) \land \tag*{(4.4)}\\
	&\ltlg\left(\boxed{\#}_{\pi_i} \Leftrightarrow \left((;_{\pi_i} \Rightarrow \neg b_{\pi_i}) \mathcal{U} (\ltlnext \#_{\pi_i})\right) \right) \land \tag*{(5)}\\
	&\mathit{exactlyOnce}(\mathit{error}_{\pi_i}) \land \mathit{exactlyOnce}(\mathit{errorStart}_{\pi_{i-1}})  \Rightarrow \mathit{refuteError}_i \tag*{(6)}
\end{align}
As before (1) specifies that the counter starts correctly.
(2) again specifies the spacing of $;$.
Note how we use that yardstick to declare a spacing of $\tilde{n}$. As $\pi_{i-1}$ is resolved after $\pi_i$ and the $\exists$-player can not control when a counter starts, it must holds for all possible starting points of the counter.
Here we use the helper $\mathit{next}(\psi_1, \psi_2)$ which specifies that the very next time $\psi_1$ holds $\psi_2$ holds as well.
In (2) we thus express that if the yardstick on $\pi_{i-1}$ aligns with a $;$ then a $;$ occurs the next time the yardstick holds. $;$ thus occurs exactly every $\tilde{n}$ many steps. 
(3) specifies that $\#$ occurs exactly if the $a$-counter is zero. Together with (1) this implies a start of the counter at $0$.  
(4) specifies the correctness of the $a$-counter. 
Unlike, in the case for $i = 1$ we split this statement into (4.1) to (4.4) for each possible combination of $a$ holding at the current position and the position $\tilde{n}$ steps apart. 
If they agree (i.e., $a$ holds now and in $\tilde{n}$ steps or does not hold now and does not hold in $\tilde{n}$ steps) then a $0$ (i.e, $\neg a$) muss occur before the end of the counter. Otherwise (in cases 4.3 and 4.4) all successor bits must be $1$. 
Note again that as we consider every possible counter start on $\pi_{i-1}$ all positions of the $a$-counter must be correct. 
(5) specifies that the $\boxed{\#}$ symbol agree exactly when the next $b$-count is $0$. 

The tricky bit is now the refute error statement, i.e., the verification of the $b$-counter in (6). We only want to verify this when the $\forall$-player claims to have found an error, i.e., he must have set $\mathit{error}$ proposition on $\pi_i$ and also set the $\mathit{errorStart}$ proposition at $\pi_{i-1}$.
Note that by construction of the CGS, the occurrence of $\mathit{errorStart}$ on $\pi_{i-1}$ implies that the $\exists$-player can determine when to start his counter/yardstick. 
In the formula we will require that he must start his counter while still in the current $a$-count pointed to by $\mathit{errorStart}$.
We define
\begin{align*}
	\mathit{refuteError}_i := &\ltlg ( \mathit{errorStart}_{\pi_{i-1}} \Rightarrow \mathit{before}(;_{\pi_i}, \boxed{\#}_{\pi_{i-1}}) ) \land \tag*{(1)}\\
	&(\mathit{countCorrect}_i \lor \mathit{noValidPos}_i) \tag*{(2)}
\end{align*}
Here (1) requires that the $\exists$-player starts the counter soon enough. In the CGS whenever the $\forall$-player signals $\mathit{errorStart}$ the $\exists$-player can choose when the counter should start.
We want the start of the counter to signal the position at which the $a$-count at the positions $\mathit{error}_{\pi_i}$ and $\mathit{errorStart}_{\pi_{i-1}}$ do not align (see details below).
We thus want the $\exists$-player to start the counter while still in the $a$-count directly after  $\mathit{errorStart}_{\pi_{i-1}}$.
(1) now requires the counter to start before $;_{\pi_i}$ holds, i.e., the current $a$-count ends. 
(2) states that the positions pointed to by $\mathit{error}_{\pi_i}$ and $\mathit{errorStart}_{\pi_{i-1}}$ are no error by the $\exists$-player.
This can be for two reasons. Either the $b$ counter at those positions is correct, or the positions chosen by the $\forall$-player was not valid, i.e., did not point to the same position at two consecutive $b$-counts. 

We begin with $\mathit{countCorrect}_i$ which indicates that the $\exists$-player did not make a mistake at the positions marked by $\mathit{error}_{\pi_i}$ and $\mathit{errorStart}_{\pi_{i-1}}$.
Formally
\begin{align*}
	\mathit{countCorrect}_i := &\left(\ltlg (\mathit{errorStart}_{\pi_{i-1}} \Rightarrow b_{\pi_i}) \Leftrightarrow  \ltlg (\mathit{error}_{\pi_{i}} \Rightarrow b_{\pi_i})\right) \Leftrightarrow \\
	&\left(\ltlg (\mathit{errorStart}_{\pi_{i-1}} \Rightarrow   \mathit{before}(;_{\pi_i}, \neg b_{\pi_i}))\right)
\end{align*}
i.e., the bits at the two positions agree if and only if there is a $0$ bit after the earlier positions (pointed to by $\mathit{errorStart}_{\pi_{i-1}}$).

The second alternative is, that the $\forall$-player did not place the $\mathit{errorStart}_{\pi_{i-1}}$ at the correct location, i.e., the position of $\mathit{errorStart}_{\pi_{i-1}}$ and $\mathit{error}_{\pi_i}$ are not the same positions in two consecutive $b$-counts.
This is exactly the case when they either differ by more than one $\boxed{\#}$-symbol or the $a$-counter at both positions are not the same.
We now define $\mathit{noValidPos}_i$ to hold if either of those cases occur.
{\small\begin{align*}
		\mathit{noValidPos}_i := &\ltlg (\mathit{errorStart}_{\pi_{i-1}} \Rightarrow \neg \mathit{onceBefore}(\#_{\pi_i}, \mathit{error}_{\pi_i})  ) \tag*{(1)} \lor \\
		&\ltlg (\mathit{errorStart}_{\pi_{i-1}} \Rightarrow \mathit{next}(\boxed{\#}_{\pi_{i-1}}, a_{\pi_{i}}) ) \not\Leftrightarrow \ltlg (\mathit{error}_{\pi_i} \Rightarrow \mathit{next}(\boxed{\#}_{\pi_{i-1}}, a_{\pi_{i}}) ) \tag*{(2)}
\end{align*}}%
(1) encodes that $\mathit{errorStart}_{\pi_{i-1}}$ and $\mathit{error}_{\pi_i}$ are not in two consecutive $b$-counts, i.e, the separator $\#$ occurs more than once or not a single time between the positions of $\mathit{error}_{\pi_i}$ and $\mathit{errorStart}_{\pi_{i-1}}$.
(2) encodes that the $a$-count at the two positions does not match. Recall that the $\exists$-player can choose the start of counter (if the $\forall$-player marked a position via $\mathit{errorStart}$). We require that the counter was started at exactly the positions that does not match in the $a$ counter. 
As the yardstick (counter) on $\pi_{i-1}$ is exactly of length $\tilde{n}$ (2) compares the same position in the two $a$-counts. 
The first time the yardstick holds after the $\mathit{errorStart}_{\pi_{i-1}}$ and the first time after $\mathit{error}_{\pi_i}$, the $a$ value on $\pi_i$ should thus not agree.

\subsection{Lower Bound Proof}

\newtheorem*{relb}{Restatement of Theorem \ref{theo:lb}}
\begin{relb}
	Model checking of a linear \HyperATLS~formula with $k$ complex and $l$ simple quantifiers is $(2k + l - 1)$-\EXPSPACE-hard in the size of the formula, provided $l \geq 1$.
\end{relb}
\begin{proof}
	In the previous construction we constructed a formula of the form  $\llangle A_1 \rrangle \pi_1. \cdots \llangle A_k \rrangle \pi_k. \psi$ can encode a yardstick of length $\mathcal{T}(2k, n)$ for some fixed $n$. 
	A formula of the form $\exists \pi. \llangle A_1 \rrangle \pi_1. \cdots \llangle A_k \rrangle \pi_k. \psi$ can now encode the termination of a $\mathcal{T}(2k, n)$-space bounded Turing maschine by follwoing the proof of \cite{DBLP:phd/dnb/Rabe16}.
	Intuitively, the existential quantifier enumerates configurations of the Turing machine and we use the yardstick to verify the correctness of two consecutive configurations (see \cite[Lemma 5.6.3]{DBLP:phd/dnb/Rabe16} for details).
	
	In case $l > 1$ we can combine our counter approach with the simpler construction for \HyperLTL~given in \cite{DBLP:phd/dnb/Rabe16}.
	In particular with $k$-complex and $l$ simple quantifier we encode a yardstick of length $\mathcal{T}(2k+l-1, n)$ by extending our doubly exponential construction (for complex quantifier) with the construction from \cite{DBLP:phd/dnb/Rabe16} for simple quantification.
	(Note that in \HyperLTL, $l$ (simple) quantifiers can encode a yardstick of $l-1$ exponents).
	As in \cite{DBLP:phd/dnb/Rabe16} the resulting formula begins with an existential quantifier and can therefore encode the existential quantifier used to encode the Turing machine.
\end{proof}

We note that \refTheo{lb} does not give a satisfactory lower bound for the case where $l = 0$.  
The best lower bound we can currently show is $(2k - 2)$-\EXPSPACE~by simply regarding one of the complex quantifiers as simple. 
The exact results we know so far (in the size of the formula), are summarized in \refTable{uplow}.

\begin{table}
	
	\begin{tabular}{|c|c|c|c|}
		\hline
		& $k = 0$ & $k=1$ & $k>1$\\
		\hline
		$l=0$ & \xmark & $2$-\EXPTIME-complete \cite{DBLP:journals/jacm/AlurHK02} & \thead{$(2k-2)$-\EXPSPACE-hard\\ in $2k$-\EXPTIME}\\
		\hline
		$l > 0$ & $(l-1)$-\EXPSPACE-complete \cite{DBLP:conf/cav/FinkbeinerRS15} & \multicolumn{2}{c|}{$(2k+l-1)$-\EXPSPACE-complete}\\
		\hline
	\end{tabular}
	\vspace{0.2cm}
	\caption{Upper and lower bounds on the complexity of \HyperATLS~model checking with $k$ complex and $l$ simple quantifiers in the size of the formula.
	The \xmark~indicates that there is no formula within this fragment (i.e,, every \HyperATLS-formula has at least on quantifier).
	}\label{tab:uplow}
\end{table}

\section{Additional Material for \refSection{rel}}

\newtheorem*{realtl}{Restatement of Theorem \ref{theo:altl}}
\begin{realtl}
	For any Kripke structure $\mathcal{G}$ and \AHLTL~formula of the form $\forall \pi_1. \cdots \forall \pi_n. \mathbf{E}. \varphi$ it holds that if $\mathcal{G}_\mathit{stut} \models [\llangle \mathit{sched} \rrangle \pi_1. \cdots. \llangle \mathit{sched} \rrangle \pi_n]~\varphi \land \bigwedge_{i \in \{1, \cdots, n\}} \mathit{fair}_{\pi_i}$ \textbf{(1)} then $\mathcal{G} \models_{\texttt{AHLTL}} \forall \pi_1. \cdots \forall \pi_n. \mathbf{E}. \varphi$ \textbf{(2)}.
	If $\varphi$ is an admissible formula, \textbf{(1)} and \textbf{(2)} are equivalent.
\end{realtl}
\begin{proof}
	We first consider the first direction. Let 
	$$\mathcal{G}_\mathit{st} \models [\llangle \mathit{sched} \rrangle \pi_1. \cdots. \llangle \mathit{sched} \rrangle \pi_n]~\varphi \land \bigwedge_{i \in \{1, \cdots, n\}} \mathit{fair}_{\pi_i}$$%
	 and let $f_\mathit{sched}^i$ for $1 \leq i \leq n$ be the winning strategy for the scheduler.
	We construct a trajectory as follows: Assume fixed paths $t_1, \cdots t_n$. For each timestep $m \in \mathbb{N}$ we query $f^i_\mathit{sched}(\mathit{zip}(t_1, \cdots, t_n)[0, m])$. If the strategy chooses to stutter, traces $j$ does not progress in timestep $m$. 
	The resulting trajectory satisfies te formula as the strategy for $\mathit{sched}$ is wining. Due to the additional constraint $\bigwedge_{i \in \{1, \cdots, n\}} \mathit{fair}_{\pi_i}$, a system can not be stuttered forever so the resulting trajectory is fair (as required in the \AHLTL~semantics \cite{hyperaltl}).
	So $\mathcal{G} \models_{\texttt{AHLTL}} \forall \pi_1. \cdots \forall \pi_n. \mathbf{E}. \varphi$.

	Now lets show the other direction. Assume $\mathcal{G} \models_{\texttt{AHLTL}} \forall \pi_1. \cdots \forall \pi_n. \mathbf{E}. \varphi$ where $\varphi$ is admissible. $\varphi$ thus only consist of stutter invariant properties on a single trace (which as satisfied by any trace in the system) and a conjunction of formulas of the from $\bigwedge_{a \in P} a_{\pi_i} \leftrightarrow a_{\pi_j}$ for a set of atomic propositions $P$.
	We claim that in case there exists a trivial wining strategy. Due to the constraints on admissible formulas (which must hold at each timepoint) the future traces are irrelevant to fix the stuttering. 
	The winning strategy for the stuttering player in each copy is thus the following: Check if there exists any way to progress some of the traces without violating the property. If this is possible, allow a proper step. Otherwise stutter.
	The resulting strategy thus stutters a path until the earliest possible time point at which a progress is possible. 
	The proof of \cite{hyperaltl} identified this point via a \LTL formula in their reduction to synchronous \HyperLTL.
	The crucial observation is that if it is \emph{possible} to advance the system, it is always a dominant strategy to advance it and not wait longer. 
	As $\mathcal{G} \models_{\texttt{AHLTL}} \forall \pi_1. \cdots \forall \pi_n. \mathbf{E}. \varphi$ the resulting strategy is winning, so $\mathcal{G}_\mathit{st} \models [\llangle \mathit{sched} \rrangle \pi_1. \cdots. \llangle \mathit{sched} \rrangle \pi_n]~\varphi \land \bigwedge_{i \in \{1, \cdots, n\}} \mathit{fair}_{\pi_i}$ as required.
\end{proof}

\section{Details on Prototype Implementation}

In this section, we discuss our prototype implementation.
We structures this section as follows:
We begin by showing how we can reduce a MSCGS and an \HyperATLS~formula of the from $[\llangle A_1 \rrangle \pi_1. \cdots \llangle A_k \rrangle \pi_k]~\varphi$ where $\varphi$ is quantifier free (i.e., an \LTL~formula) to a parity game.
This avoids the construction and emptiness check of an alternating automata.

Our tool, can, in theory, check arbitrary MSCGS. As a more convenient formalism we constructed a frontend that accepts a programs and automatically translates it into a MSCGS which can the be checked. 
We present a simple imperative programming language over bit vectors and present a simple semantics into a CGS. 

\subsection{Parity Games}

A parity game is a tuple $\mathit{PG} = (V, E, v_0, p, c)$ where $V$ is a set of vertices $E \subseteq V \times V$ as set of edges, $v_0$ a dedicated initial vertex (this is omitted in most definitions in the literature), $p : V \to \{0, 1\}$ assigning each vertex to a player $0$ or player $1$ and $c : V \to \mathbb{N}$ labels each vertex with a colour. 
We assume that for every $v \in V$ there is a $v' \in V$ with $(v, v') \in E$.
We write $V_0 := V \cap p^{-1}(0)$ and $V_1 := V \cap p^{-1}(1)$ for the vertices controlled by player $0$ and $1$ respectively.

A strategy for a player now maps finite histories to successor states.
For states in $V_0$, player $0$ can choose the successor and similarly for states in $V_1$.
It is known that parity games admit positional strategies, i.e., a strategy can ignore the history of a play and decide on a successor based purely on the current vertex \cite{DBLP:journals/tcs/Zielonka98}.
Furthermore parity games are determined, i.e, for each vertex there exists a unique winning player \cite{borel}.
We therefore simply introduce positional strategies.

A positional strategy for player $l \in \{0, 1\}$ is mapping $f_l : V_l \to V$ such that for each $v \in V_l$, $(v, f_l(v)) \in E$.
Given a strategy $f_0$ for player $0$ and a strategy $f_1$ for player $1$ and a vertex $v$ we define $\mathit{path}(v,f_0, f_1) \in V^\omega$ as the unique run in $u_0u_1\cdots$ where $u(i) = v$ and for each $i \in \mathbb{N}$ we have that $u(i+1) \in f_l(u(i))$ where $l = p(u(i))$.
We call $u$ winning for player $0$ if the smallest priority occurring infinitely many times in even. 
We define $W_0$ as all vertices $v$ such that player $0$ wins (i.e., there exists a $f_0$ such that  for all $f_1$, the path $\mathit{path}(v,f_0, f_1)$ is winning for player $0$). We define $W_1 := V \setminus W_0$ (Recall that parity games are determined).
We say that player $0$ wins a parity game $\mathit{PG} = (V, E, v_0, p, c)$ if $v_0 \in W_0$.

\subsection{Reduction to Parity Games}

Assume we have a \HyperATLS~formula in the parallel-composition fragment.
We consider the general case where the formula uses extended path quantification and annotates quantifiers with game structures.
Let $\varphi = [\llangle A_1 \rrangle_{\mathcal{G}_1} \, \pi_1. \cdots \llangle A_n \rrangle_{\mathcal{G}_n} \, \pi_n]. \psi$ where $\mathcal{G}_i$ are MSCGSs and $\psi$ is quantifier free (i.e., a \LTL~formula).

Let $\mathcal{G}_i = (S_i, s_{i, 0}, \agents_i, \mathscr{M}_i, \delta_i, d_i, \atomic_i, L_i)$.
We assume that $\atomic_i \cap \atomic_j = \emptyset$ for all $i \neq j$.
Write $\Sigma := S_1 \times \cdots \times S_n$.
We begin by translating the \LTL~formula $\psi$ to a deterministic parity automaton $\mathcal{A}_\psi = (Q, q_0, \Sigma, \rho, c)$ where $\rho : Q \times \Sigma \to Q$, is the a deterministic transition function. 
We now describe the parity game construction. The stages in the MSCGS are naturally encoded as states in a parity game that model the incremental move selection.
Nodes in which agents controlled by a strategy (the ones in one of the $A_i$) select their moves are then controlled by player $0$ and other nodes by player $1$.
We describe the components of the resulting PG individually:

\subparagraph{Vertices}
The set of vertices is defined via 
$$V := (Q \times \prod_{i=1}^{n} S_i) \cup (Q \times \prod_{i=1}^{n} S_i \times \prod_{i=1}^{n} (\agents_i \rightharpoonup \mathscr{M}_i) \times \mathbb{N} \times \mathbb{B})$$
The initial vertex is $(q_0, s_{1, 0}, \cdots, s_{n, 0})$.
A vertex in $(Q \times \prod_{i=1}^{n} S_i \times \prod_{i=1}^{n} \agent_i \to \mathscr{M}_i \times \mathbb{N} \times \mathbb{B})$ has the from $(q, s_1, \cdots, s_n, \sigma_1, \cdots, \sigma_n, l, b)$. Here $q$ is the current automaton state, $s_1, \cdots, s_n$ the current state of the self-composed system, $\sigma_1, \cdots, \sigma_n$ are the partial assignments chosen so far (corresponding to stages of the MSCGS), $l$ is the current state and $b$ indicates if its current the turn of the agents controlled by a strategy to choose a move or the agents that are adversarial.
A node of the from $(q, s_1, \cdots, s_n, \sigma_1, \cdots, \sigma_n, l, b)$ is thus controlled by player $b$.
Nodes in $ (Q \times \prod_{i=1}^{n} S_i)$ can be controlled by either player as they have a unique outgoing edge (see below).
The colour of a node  $(q, s_1, \cdots, s_n, \sigma_1, \cdots, \sigma_n, l, b)$ or $(q, s_1, \cdots, s_n)$ is given by $c(q)$, i.e., the colouring of state $q$ in the DPA.

\subparagraph{Edges}
For a set of agents $A \subseteq  \agents_i$ and a level $l$ we define $\mathit{level}_i(A, l) := A \cap d^{-1}(l)$.
Edges in the parity game have the following form:

\begin{itemize}
	\item There is a edge from $(q, s_1, \cdots, s_n)$ to $(q', s_1, \cdots, s_n, \emptyset, \cdots, \emptyset, 0, \top)$ whenever $q' = \rho(q, (s_1, \cdots, s_n))$.
	Here $\emptyset$ denotes the partial move vector with empty domain.
	Such transitions model a step of the automaton and indicate that selection of moves begins in stage $0$.

	\item There is an edge from $(q, s_1, \cdots, s_n, \sigma_1, \cdots, \sigma_n, l, b)$ to $(q, s'_1, \cdots, s'_n)$ whenever $\sigma_1, \cdots, \sigma_n$ are complete (i.e., the domain of $\sigma_i$ is $\agents_i$) and $s'_i = \delta_i(s_i, \sigma_i)$ for all $i$.
	These transition model step of the system after the move vector was fixed for all agents in all copies.

	\item There is an edge from $(q, s_1, \cdots, s_n, \sigma_1, \cdots, \sigma_n, l, b)$ to $(q, s_1, \cdots, s_n, \sigma'_1, \cdots, \sigma'_n, l', \neg b)$ whenever some of the $\sigma_i$ is not total.
	If $b = \top$ the strategy profile should be updated by agents in $A$ in stage $l$. If $b = \bot$ the counteragents (the ones not in $A$) in stage $l$ choose their move.
	$\neg b$ negates the flipped value of $b$.
	
	\begin{itemize}
		\item If $b = \top$ we require that $\sigma'_i$ is obtained from $\sigma_i$ by fixing moves for all agents in $\mathit{level}_i(A_i, l)$. Formally $\mathit{dom}(\sigma'_i) = \mathit{dom}(\sigma_i) \cup \mathit{level}_i(A_i, l)$ and $\sigma_i =_{\mathit{level}_i(A_i, l)} \sigma'_i$. Here $f =_C g$ means that $f$ and $g$ agree on all inputs that are \emph{not} in $C$.
		In this case $l' = l$.
		
		\item If $b = \top$ we require that $\sigma'_i$ is obtained from $\sigma_i$ by fixing moves for all agents in $\mathit{level}_i(\overline{A_i}, l)$. Formally $\mathit{dom}(\sigma'_i) = \mathit{dom}(\sigma_i) \cup \mathit{level}_i(\overline{A_i}, l)$ and $\sigma_i =_{\mathit{level}_i(\overline{A_i}, l)} \sigma'_i$.
		In this case $l' = l + 1$.
	\end{itemize}
\end{itemize}

The edge relation thus requires the agents to fix their moves incrementally.
Initially we have a node of the form $(q', s_1, \cdots, s_n, \emptyset, \cdots, \emptyset, 0, \top)$. Thus the agents in stage $0$ contained in the agents set (as $b = \top$) fix their moves.
After they have fixed their move, the agents not included on the strategy set choose their move. We then increase the stage by one, so the agent controlled by a strategy in stage $1$ choose their moves. This process continuous until all agents fixed their move in which case we can compute the successor state of the system.
As all agents belong to some stage, the incremental move selection is finite, so a proper transition of the system is made eventually.

Note that the state space of the game is infinite as the stage counter is unbounded. As there must be a maximal stage at which all agents made their move, the fragment reachable from the fixed initial state is, however, finite.
Note that the size of the reachable fragment is polynomial in the size of the CGS and exponential in the number of agents.
It is easy to see that every winning strategy for player $0$ directly corresponds to the a winning strategy in the \HyperATLS~semantics.
We therefore get:

\begin{proposition}
	For any MSCGS $\mathcal{G}_1, \cdots, \mathcal{G}_n$ and $\varphi = [\llangle A_1 \rrangle_{\mathcal{G}_1} \, \pi_1. \cdots \llangle A_n \rrangle_{\mathcal{G}_n} \, \pi_n]~\psi$ where $\psi$ is quantifier free, we have that $\varphi$ is satisfied if any only if the parity game defined above is won by player $0$.
\end{proposition}

\subsection{Simple Imperative Language}

We consider a simple imperative programming language that corresponds to the language supported by our model checker.
We define a canonical operational semantics from a program into a CGS, which can afterwards be checked against a \HyperATLS~formula. 

Fix a finite set of variables $\mathcal{V}$ and a domain function $z : \mathcal{V} \to \mathbb{N}_{\geq 1}$ giving the bitwidth of each variable.
Expressions are defined by:
$$
e := x \mid \mathit{true} \mid \mathit{false} \mid e_1 \land e_2 \mid e_1 \lor e_2 \mid \neg e \mid e_1 @ e_2 \mid e[n]
$$
where $x \in \mathcal{V}$ is a variable and $n \in \mathbb{N}$ an index used as an projection.
$\land, \lor$ and $\neg$ denote bitwise and or and negation (operators should have the same bit width). $@$ denotes concatenation of two values and $e[n]$ the projection of the $n$ bit.
A variable state is function $\sigma: \mathcal{V} \to \mathbb{B}^+$ such that for each $x \in \mathcal{V}$, $|\sigma(x)| = z(x)$.
Here $\mathbb{B} = \{\top,\bot\}$.
Given an expression $e$ and a variable state $\sigma \to \mathbb{B}^+$, we define $\llbracket e \rrbracket(\sigma) \in \mathbb{B}^+$ as follows:\\
\begin{minipage}{0.5\textwidth}
	\begin{align*}
		\llbracket x \rrbracket(\sigma) &= \sigma(x)\\
		\llbracket \mathit{true} \rrbracket(\sigma) &= [\top]\\
		\llbracket \mathit{false} \rrbracket(\sigma) &= [\bot]\\
		\llbracket \neg e \rrbracket(\sigma) &= ! \llbracket e \rrbracket(\sigma)
	\end{align*}
\end{minipage}%
\begin{minipage}{0.5\textwidth}
	\begin{align*}
		\llbracket e_1 \land e_2 \rrbracket(\sigma) &= \llbracket e_1 \rrbracket(\sigma) \&\& \llbracket e_2 \rrbracket(\sigma)\\
		\llbracket e_1 \lor e_2 \rrbracket(\sigma) &= \llbracket e_1 \rrbracket(\sigma) || \llbracket e_2 \rrbracket(\sigma)\\
		\llbracket e_1 @ e_2 \rrbracket(\sigma) &= \llbracket e_1 \rrbracket(\sigma) \cdot \llbracket e_2 \rrbracket(\sigma)\\
		\llbracket e[n] \rrbracket(\sigma) &= \llbracket e \rrbracket(\sigma)(n)
	\end{align*}
\end{minipage}

\vspace{0.2cm}

\noindent where $\&\&, ||, !, \cdot$ are bitwise and, bitwise or (only defined on operands of the same size), bitwise negation and string concatenation. 
Programs are then define via:
\begin{align*}
	P := x := e &\mid x := \texttt{Read}_H \mid x := \texttt{Read}_L \mid \texttt{if}~e~\texttt{then}~P_1~\texttt{else}~P_2 \\
	&\mid \texttt{if}~*~\texttt{then}~P_1~\texttt{else}~P_2 \mid \texttt{while} (e)~P \mid P_1 ; P_2 \mid \lightning
\end{align*}%
where $e$ is an expression.
$\lightning$ indicates a terminated program.
We assume the program to be well formed in the usual sense. 
For instance:  \textbf{(1)} operations are only performed on operants with compatible bit lengths, \textbf{(2)} for conditional statements the guarding expression is of length $1$.

The semantics of a program is as expected for assignments, conditionals, loops and sequential composition.
The \texttt{Read} construct allows to read inputs from a high-security or low-security sources. 
$\texttt{if} * \texttt{then}~P_1~\texttt{else}~P_2$ is a non-deterministic branching between both programs.
We formalize the semantics of a program via a relation $\to$ that maps pairs of program and variable state to pairs of program and variable state.
$\langle P, \sigma \rangle \to \langle P', \sigma' \rangle$ means that program $P$ on state $\sigma$ does, in one step, reduce to $P'$ and changes the state to $\sigma'$.
For most language constructs this relation is deterministic. The exceptions are: \textbf{(1) }for $\texttt{if}~*~\texttt{then}~P_1~\texttt{else}~P_2$ both programs are possibly successors, and \textbf{(2)} a read statement causes all possible values of the read variable to be possible successors states. 
The semantics is formalized in \refFig{sem} (Recall that $z$ gives the bit-width of each variable in use).

\begin{figure}
	\begin{minipage}{0.5\textwidth}
		\begin{prooftree}
			\AxiomC{$\sigma' = \sigma[x \mapsto \llbracket e \rrbracket(\sigma)]$}
			\UnaryInfC{$\langle x:= e, \sigma \rangle \to \langle \lightning, \sigma' \rangle$}
		\end{prooftree}
	\end{minipage}
	\begin{minipage}{0.5\textwidth}
		\begin{prooftree}
			\AxiomC{$\sigma' = \sigma[x \mapsto b]$}
			\AxiomC{$b \in \mathbb{B}^{z(x)}$}
			\AxiomC{$P \in \{L, H\}$}
			\TrinaryInfC{$\langle x := \texttt{Read}_P, \sigma \rangle \to \langle \lightning, \sigma' \rangle$}
		\end{prooftree}
	\end{minipage}
	
	\vspace{0.3cm}
	
	\begin{minipage}{0.5\textwidth}
		\begin{prooftree}
			\AxiomC{$\llbracket e \rrbracket(\sigma) = [\top]$}
			\UnaryInfC{$\langle \texttt{if}~e~\texttt{then}~P_1~\texttt{else}~P_2, \sigma \rangle \to \langle P_1, \sigma' \rangle$}
		\end{prooftree}
	\end{minipage}%
	\begin{minipage}{0.5\textwidth}
		\begin{prooftree}
			\AxiomC{$\llbracket e \rrbracket(\sigma) = [\bot]$}
			\UnaryInfC{$\langle \texttt{if}~e~\texttt{then}~P_1~\texttt{else}~P_2, \sigma \rangle \to \langle P_2, \sigma' \rangle$}
		\end{prooftree}
	\end{minipage}
	
	\vspace{0.3cm}
	
	\begin{minipage}{0.5\textwidth}
		\begin{prooftree}
			\AxiomC{}
			\UnaryInfC{$\langle \texttt{if}~*~\texttt{then}~P_1~\texttt{else}~P_2, \sigma \rangle \to \langle P_1, \sigma\rangle$}
		\end{prooftree}
	\end{minipage}%
	\begin{minipage}{0.5\textwidth}
		\begin{prooftree}
			\AxiomC{}
			\UnaryInfC{$\langle \texttt{if}~*~\texttt{then}~P_1~\texttt{else}~P_2, \sigma \rangle \to \langle P_2, \sigma \rangle$}
		\end{prooftree}
	\end{minipage}
	
	\vspace{0.3cm}
	
	\begin{minipage}{0.5\textwidth}
		\begin{prooftree}
			\AxiomC{$\llbracket e \rrbracket(\sigma) = [\bot]$}
			\UnaryInfC{$\langle \texttt{while}(e)~P, \sigma \rangle \to \langle \lightning, \sigma \rangle$}
		\end{prooftree}
	\end{minipage}%
	\begin{minipage}{0.5\textwidth}
		\begin{prooftree}
			\AxiomC{$\llbracket e \rrbracket(\sigma) = [\bot]$}
			\UnaryInfC{$\langle \texttt{while}(e)~P, \sigma \rangle \to \langle P;~\texttt{while}(e)~P, \sigma \rangle$}
		\end{prooftree}
	\end{minipage}
	
	\vspace{0.3cm}
	
	\begin{minipage}{0.33\textwidth}
		\begin{prooftree}
			\AxiomC{$\langle P_1, \sigma \rangle \to \langle \lightning, \sigma' \rangle$}
			\UnaryInfC{$\langle P_1 ; P_2, \sigma \rangle \to \langle P_2, \sigma' \rangle$}
		\end{prooftree}
	\end{minipage}%
	\begin{minipage}{0.33\textwidth}
		\begin{prooftree}
			\AxiomC{$\langle P_1 ; P_2, \sigma \rangle \to \langle P_1', \sigma' \rangle$}
			\AxiomC{$P_1' \neq \lightning$}
			\BinaryInfC{$\langle P_1 ; P_2, \sigma \rangle \to \langle P_1';P_2, \sigma' \rangle$}
		\end{prooftree}
	\end{minipage}%
	\begin{minipage}{0.33\textwidth}
		\begin{prooftree}
			\AxiomC{}
			\UnaryInfC{$\langle \lightning, \sigma \rangle \to \langle \lightning, \sigma \rangle$}
		\end{prooftree}
	\end{minipage}
	
	\caption{Small-step semantics of our imperative programming language.} \label{fig:sem}
\end{figure}

We also associate each program construct to a player from $\{\agent_N, \agent_H, \agent_L\}$. In our simple language the programs are complied into a CGS where in each state exactly one of the players has a decision and all other are irrelevant. 
The resulting CGS is therefore turn-based in the sense of \cite{DBLP:journals/jacm/AlurHK02}.
This mapping $\mathit{player} : P \to  \{\agent_N, \agent_H, \agent_L\}$ is inductively defined via:

\begin{minipage}{0.4\textwidth}
	\begin{align*}
		\mathit{player}(x := e) &= \agent_N\\
		\mathit{player}(x := \texttt{Read}_H) &= \agent_H\\
		\mathit{player}(x := \texttt{Read}_L) &= \agent_L\\
		\mathit{player}( \texttt{while} (e) P) &= \agent_N
	\end{align*}
\end{minipage}%
\begin{minipage}{0.6\textwidth}
	\begin{align*}
		\mathit{player}( \texttt{if}~e~\texttt{then}~P_1~\texttt{else}~P_2) &= \agent_N\\
		\mathit{player}( \texttt{if}~*~\texttt{then}~P_1~\texttt{else}~P_2) &= \agent_N\\
		\mathit{player}( P_1 ; P_2) &= \mathit{player}(P_1)\\
		\mathit{player}( \lightning) &= \agent_N
	\end{align*}
\end{minipage}

\vspace{0.2cm}
\noindent Note that only the cases for reading and non-deterministic branching are relevant, as in all other cases there exists a unique successor configuration.

\subparagraph{The final CGS}
Given a program $P_0$ and a set of variables $\mathcal{V}$ (including all variables used in $P$) and a bitwidth $z : \mathcal{V} \to \mathbb{N}_{\geq 1}$ the states in the CGS consist of all configuration of the from $\langle P, \sigma\rangle$ where $P$ is a program and $\sigma$ a variable state conforming to $z$. 
The initial state is $\langle P_0, \mathbf{0}\rangle$ where $P_0$ is the given program and $\mathbf{0}$ the variable state that maps each variable to the $0$ vector (of length given by $z$).
The successors in a state $\langle P, \sigma \rangle$ are given by $\to$. The player $\mathit{player}(P)$ is responsible for choosing a successor.
Atomic propositions have the from $xi$ where $x \in \mathcal{V}$ is a variable and $0 \leq i < z(x)$ a position. Proposition $xi$ is set to true in state $\langle P, \sigma \rangle$ iff $\sigma(x)(i) = \top$.
This gives a canonical reduction from a program into a CGS.

\section{Additional Material for \refSection{proto} - Experimental Evaluation}

In this section we provide the example programs and properties checked in our evaluation. 
Our tool is publicly  available at GitHub: \url{https://github.com/reactive-systems/hyperatl}

\subsection{Information-Flow Control}

\begin{figure}[t]
	\small
	\begin{subfigure}[b]{0.24\textwidth}
		\myvar{o} $\leftarrow$ \myconst{$\bot$}\\
		\mycontrol{while}(\myconst{true})\\
		\makebox[0.3cm]{} \myvar{h} $\leftarrow$ \mycontrol{read}$_H$()\\
		\makebox[0.3cm]{} \textcolor{controlColor}{if} (\myvar{h}~$=$~\myconst{$\top$} ) \textcolor{controlColor}{then}\\
		\makebox[0.6cm]{} \myvar{o} $\leftarrow$ !(\myvar{o})\\
		\makebox[0.3cm]{} \mycontrol{else}\\
		\makebox[0.6cm]{} \myvar{o} $\leftarrow$ !(\myvar{o}) $\land$~(\myvar{h}~$\lor$~!(\myvar{h}))
		
		\subcaption{\texttt{P1}}
	\end{subfigure}
	\begin{subfigure}[b]{0.24\textwidth}
		\myvar{l} $\leftarrow$ \myconst{$\bot$}\\
		\myvar{o} $\leftarrow$ \myconst{$\top$}\\
		\mycontrol{while}(\myconst{true})\\
		\makebox[0.3cm]{} \myvar{h} $\leftarrow$ \mycontrol{read}$_H$()\\
		\makebox[0.3cm]{} \myvar{b} $\leftarrow$ \myvar{l}\\
		\makebox[0.3cm]{} \myvar{l} $\leftarrow$ \mycontrol{read}$_L$()\\
		\makebox[0.3cm]{} \textcolor{controlColor}{if} (\myvar{h}~$=$~\myconst{$\top$} ) \textcolor{controlColor}{then}\\
		\makebox[0.6cm]{} \myvar{o} $\leftarrow$ \myvar{l}~$\lor$~\myvar{b}\\
		\makebox[0.3cm]{} \mycontrol{else}\\
		\makebox[0.6cm]{} \myvar{o} $\leftarrow$ \myvar{l}~$\lor$~\myvar{b}\\
		
		\subcaption{\texttt{P2}}
	\end{subfigure}
	\begin{subfigure}[b]{0.24\textwidth}
		\myvar{l} $\leftarrow$ \myconst{$\bot$}\\
		\myvar{o} $\leftarrow$ \myconst{$\top$}\\
		\mycontrol{while}(\myconst{true})\\
		\makebox[0.3cm]{} \myvar{h} $\leftarrow$ \mycontrol{read}$_H$()\\
		\makebox[0.3cm]{} \textcolor{controlColor}{if} ($\star$) \textcolor{controlColor}{then}\\
		\makebox[0.6cm]{} \myvar{o} $\leftarrow$ \myconst{$\top$}\\
		\makebox[0.3cm]{} \mycontrol{else}\\
		\makebox[0.6cm]{} \myvar{o} $\leftarrow$ \myconst{$\bot$}\\
		
		\subcaption{\texttt{P3}}
	\end{subfigure}
	\begin{subfigure}[b]{0.24\textwidth}
		\myvar{o} $\leftarrow$ \myconst{$\top$}\\
		\mycontrol{while}(\myconst{true})\\
		\makebox[0.3cm]{} \textcolor{controlColor}{if} ($\star$) \textcolor{controlColor}{then}\\
		\makebox[0.6cm]{} \myvar{h} $\leftarrow$ \mycontrol{read}$_H$()\\
		\makebox[0.6cm]{} \textcolor{controlColor}{if} (\myvar{h}~$=$~\myconst{$\top$}) \textcolor{controlColor}{then}\\
		\makebox[0.9cm]{} \myvar{o} $\leftarrow$ \myconst{$\top$}\\
		\makebox[0.6cm]{} \mycontrol{else}\\
		\makebox[0.9cm]{} \myvar{o} $\leftarrow$ \myconst{$\bot$}\\
		\makebox[0.3cm]{} \mycontrol{else}\\
		\makebox[0.6cm]{} \textcolor{controlColor}{if} (\myvar{h}~$=$~\myconst{$\top$}) \textcolor{controlColor}{then}\\
		\makebox[0.9cm]{} \myvar{o} $\leftarrow$ \myconst{$\bot$}\\
		\makebox[0.6cm]{} \mycontrol{else}\\
		\makebox[0.9cm]{} \myvar{o} $\leftarrow$ \myconst{$\top$}\\
		
		\subcaption{\texttt{P4}}
	\end{subfigure}

	\caption{Example programs to demonstrate information-flow control. } \label{fig:exProfgrams}
\end{figure}

The example benchmark we tested in \refTable{res1} are given in \refFig{exProfgrams}. 
Note that the bitwidth of each variable can be varied, causing a state-space explosion.
The formulas we checked are the following:
\begin{itemize}
	\item \textbf{(OD)}: $ [\forall \pi_1. \forall \pi_2.]~\ltlg (o_{\pi_1} \leftrightarrow o_{\pi_2})$
	\item \textbf{(NI)}: $ [\forall \pi_1. \forall \pi_2.]~(\ltlg (l_{\pi_1} \leftrightarrow l_{\pi_2})) \to \ltlg (o_{\pi_1} \leftrightarrow o_{\pi_2})$
	\item \textbf{(simSec)}: $ [\forall_\mathcal{G} \pi_1. \llangle \agent_N \rrangle_{\mathcal{G}_\mathit{shift}} \pi_2.]~(\ltlg (l_{\pi_1} \leftrightarrow \bigcirc l_{\pi_2})) \to \ltlg (o_{\pi_1} \leftrightarrow \bigcirc o_{\pi_2})$ where $\mathcal{G}_\mathit{shift}$ is the transition system shifted by one position. 
	
	\item \textbf{(sGNI)}: $ [\forall_\mathcal{G} \pi_1. \forall_\mathcal{G} \pi_2. \exists_{\mathcal{G}_{\mathit{shift}: k}} \pi_3.]~\ltlg (h_{\pi_1} \leftrightarrow \bigcirc^k h_{\pi_2}) \land \ltlg (o_{\pi_2} \leftrightarrow \bigcirc^k o_{\pi_3} \land l_{\pi_2} \leftrightarrow \bigcirc^k l_{\pi_3})$
	
	This is the standard definition of GNI extended to take the low-security input into account. We define $\mathcal{G}_{\mathit{shift}: k}$ as the system where the behaviour is shifted by $k$-positions. 
	By changing the $k$ we can thus allow the existential player a $k$-bounded view on the future of traces $\pi_1$ and $\pi_2$.
	For every $k$ this formula implies the trace based version of $\mathit{GNI}$.
	In our examples we choose $k = 3$. Note that for instance program \texttt{P4} (\refFig{exProfgrams}) is only satisfied if $k \geq 2$, as the existential player must observe the next input. 
	In particular note that correspondence to \cite{DBLP:conf/cav/CoenenFST19}, where they used a game based reading. As we mentioned in \refSection{hypteratl}, all formulas that can be checked in their approach can also be checked in our model checker.
	\textbf{(sGNI)} is a example of such a property, as it uses a game based reading of the existential quantifier. 
	Increasing the $k$-parameter (i.e., the clairvoyance of the existential player) can be seen as adding prophecy variables for the $k$-steps (see \cite{DBLP:conf/cav/CoenenFST19}).
\end{itemize}

We mention at this point that our model checker can be used for many more interesting programs and properties. We restrict to such a minimal class for space reasons. 

\subsection{Asynchronous Hyperproperties}

\begin{figure}[!t]
	
	\begin{subfigure}[b]{0.5\textwidth}
		\myvar{o} $\leftarrow$ \myconst{$\bot$}\\
		\mycontrol{while}(\myconst{true})\\
		\makebox[0.3cm]{} \myvar{h} $\leftarrow$ \mycontrol{read}$_H$()\\
		\makebox[0.3cm]{} \textcolor{controlColor}{if} (\myvar{h}~$=$~\myconst{$\top$} ) \textcolor{controlColor}{then}\\
		\makebox[0.6cm]{} \myvar{o} $\leftarrow$ !(\myvar{o})\\
		\makebox[0.3cm]{} \mycontrol{else}\\
		\makebox[0.6cm]{} \myvar{t} $\leftarrow$ !(\myvar{o})\\
		\makebox[0.6cm]{} \myvar{o} $\leftarrow$ \myvar{t}
		
		\subcaption{\texttt{Q1}}
	\end{subfigure}
	\begin{subfigure}[b]{0.5\textwidth}
		\myvar{o} $\leftarrow$ \myconst{$\bot$}\\
		\myvar{r} $\leftarrow$ \myconst{$\bot$}\\
		\mycontrol{while}(\myconst{true})\\
		\makebox[0.3cm]{} \myvar{r} $\leftarrow$ \myconst{$\top$}\\
		\makebox[0.3cm]{} \myvar{l} $\leftarrow$ \mycontrol{read}$_L$()\\
		\makebox[0.3cm]{} \myvar{r} $\leftarrow$ \myconst{$\bot$}\\
		\makebox[0.3cm]{} \textcolor{controlColor}{if} (\myvar{l}~$=$~\myconst{$\top$} ) \textcolor{controlColor}{then}\\
		\makebox[0.6cm]{} \myvar{o} $\leftarrow$ \myconst{$\top$}\\
		\makebox[0.3cm]{} \mycontrol{else}\\
		\makebox[0.6cm]{} \myvar{t} $\leftarrow$ \myconst{$\bot$}\\
		\makebox[0.6cm]{} \myvar{o} $\leftarrow$ \myvar{t}
		
		\subcaption{\texttt{Q2}}
	\end{subfigure}

	\caption{Example programs to demonstrate asynchronous hyperproperties}\label{fig:asycExample}
\end{figure}

The example programs used to verify the asynchronous hyperproperties in \refTable{res2} are given in \refFig{asycExample}.
The properties we verify are: 

\begin{itemize}
	\item \textbf{(OD)}: $ [\forall \pi_1. \forall \pi_2.]~\ltlg (o_{\pi_1} \leftrightarrow o_{\pi_2})$
	\item \textbf{(OD)$_\mathit{asynch}$}: $ [\llangle \mathit{sched} \rrangle \pi_1. \llangle \mathit{sched} \rrangle \pi_2.]~\ltlg (o_{\pi_1} \leftrightarrow o_{\pi_2}) \land \mathit{fair}_{\pi_1} \land \mathit{fair}_{\pi_2} $
	\item \textbf{(NI)$_\mathit{asynch}$}: $ [\llangle \mathit{sched} \rrangle \pi_1. \llangle \mathit{sched} \rrangle \pi_2.]~\left((\ltlg (l_{\pi_1} \leftrightarrow l_{\pi_2})) \to \ltlg (o_{\pi_1} \leftrightarrow o_{\pi_2})\right) \land \mathit{fair}_{\pi_1} \land \mathit{fair}_{\pi_2} \land \ltlg(r_{\pi_1} \leftrightarrow r_{\pi_2})$
\end{itemize}

We pay particular attention to $\textbf{(NI)}_\mathit{asynch}$, an asynchronous formulation of non-interference. Note that the implication in the formulation of non-interference (non-interference only requires the low outputs to agree in both runs if the low-inputs also agree) causes this formula not be expressible via an admissible formula in the sense of \cite{hyperaltl}.
To express a meaningful property we need to make sure that the asynchronous scheduler can not cheat by simply aligning the programs such that the input differs.
In particular the formula 
$$ [\llangle \mathit{sched} \rrangle \pi_1. \llangle \mathit{sched} \rrangle \pi_2.]~(\ltlg (l_{\pi_1} \leftrightarrow l_{\pi_2})) \to \ltlg (o_{\pi_1} \leftrightarrow o_{\pi_2})$$
is trivially satisfied.
To avoid this, we augmented the program \texttt{Q2} with a variable $r$, that indicates the reading position. We then ask the scheduler to always align the $r$ proposition and thereby also the read operations. 
Such restriction on the scheduler are e.g.~not possible in the setting of \cite{hyperaltl}.
This emphasises the point that \HyperATLS~is well suited to express many properties of interest and can easily incorporate domain knowledge in the specification (to e.g., express non-trivial properties such as non-interference).

\end{document}